\documentclass[final,leqno]{siamltex}

\usepackage{showkeys}

\usepackage{latexsym,calc,amssymb,amsfonts}
\usepackage{amsmath}
\usepackage{epsfig,setspace,color}

\def\rfl1{\rho_S^{RIP}(\delta)}

\def\rflqc1{\rho_S^{RIP}(\delta;q,C_1(\delta,\rho)\le\Upsilon)}

\def\rfl{\rho_S^{FL}(\delta)}

\def\sl1{StrongEquiv(A,\ell^1)}

\def\Lbt{{\cal L}^{BT}}
\def\Ubt{{\cal U}^{BT}}
\def\Lbct{{\cal L}^{BCT}}
\def\Ubct{{\cal U}^{BCT}}
\def\Lct{{\cal L}^{CT}}
\def\Uct{{\cal U}^{CT}}

\newcommand{\argmin}{\operatornamewithlimits{arg\,min}}

\title{Improved Bounds on Restricted Isometry \\ Constants  for Gaussian Matrices}

\author{Bubacarr Bah\thanks{School of Mathematics and Maxwell Institute, University of Edinburgh, Edinburgh, UK ({\tt b.bah@sms.ed.ac.uk}).}
 \and Jared Tanner\thanks{School of Mathematics and Maxwell Institute, University of Edinburgh, Edinburgh, UK ({\tt jared.tanner@ed.ac.uk}).  JT acknowledges support from the Leverhulme Trust.}}

\date{January 2010}

\begin{document}

\maketitle

\begin{abstract} 
The Restricted Isometry Constants (RIC) of a matrix $A$ measures how close to an isometry is the action of $A$ on vectors with few nonzero entries, measured in the $\ell^2$ norm.  Specifically, the upper and lower RIC of a matrix $A$ of size $n\times N$ is the maximum and the minimum deviation from unity (one) of the largest and smallest, respectively, square of singular values of all ${N\choose k}$ matrices formed by taking $k$ columns from $A$.  Calculation of the RIC is intractable for most matrices due to its combinatorial nature; however, many random matrices typically have bounded RIC in some range of problem sizes $(k,n,N)$.  We provide the best known bound on the RIC for Gaussian matrices, which is also the smallest known bound on the RIC for any large rectangular matrix.   Improvements over prior bounds are achieved by exploiting similarity of singular values for matrices which share a substantial number of columns.  
\end{abstract}



\pagestyle{myheadings}
\thispagestyle{plain}
\markboth{B. BAH AND J. TANNER}{Bounds on Restricted Isometry Constants: Gaussian}


\section{Introduction}\label{sec:Intro} Interest in parsimonious solutions to underdetermined systems of equations has seen a spike with the introduction of compressed sensing \cite{CompressedSensing,CTDecoding, CRTRobust}. Much of the analysis in this new topic has relied upon a new matrix quantity, the Restricted Isometry Constant (RIC), also referred to as the Restricted Isometry Property (RIP) constant. Let $A$ be a matrix of size $n\times N$ and define the set of $N$-vectors with at most $k$ nonzero entries as 
\begin{equation}
\chi^{N}(k):=\{x\in \mathbb{R}^{N}: \|x\|_{\ell^0}\leq k\}.
\end{equation}
Upper and lower RICs of $A$, $U(k,n,N;A)$ and $L(k,n,N;A)$ respectively, are defined as \cite{Cand,BCT}
\begin{equation}
\label{RIPU}
U(k,n,N;A):=\min_{c\geq 0}c \quad \mathrm{subject} ~\mathrm{to} \quad (1+c) \|x\|_{2}^{2} \geq \|Ax\|_{2}^{2} \quad\forall x \in \chi^{N}(k).
\end{equation}
\begin{equation}
\label{RIPL}
L(k,n,N;A):=\min_{c\geq 0}c \quad \mathrm{subject} ~\mathrm{to} \quad (1-c) \|x\|_{2}^{2} \leq \|Ax\|_{2}^{2}, \quad\forall x \in \chi^{N}(k);
\end{equation}

RICs differ from standard singular values squared in their combinatorial nature. $U(k,n,N;A)$ and $L(k,n,N;A)$ measure the maximum and the minimum deviation from unity (one) of the largest and smallest, respectively, square of the singular values of all $N \choose k$ submatrices of $A$ of size $n \times k$ constructed by taking $k$ columns from $A$.  The RICs can be equivalently defined as 
\begin{equation}
U(k,n,N;A) := \max_{K\subset \Omega, |K|=k} \lambda^{max}\left(A_K^* A_K\right) - 1
\end{equation}
and 
\begin{equation}
L(k,n,N;A) := 1 - \min_{K\subset \Omega, |K|=k} \lambda^{min}\left(A_K^* A_K\right)
\end{equation}
where $\Omega: = \{1,2,\dots,N\}$, $A_{K}$ is the restriction of the columns of $A$ to a support set $K \subset \Omega$ with cardinality $k$ ($|K| = k$), and $\lambda^{max}\left(B\right)$ and $\lambda^{min}\left(B\right)$ are the smallest and largest eigenvalues of $B$ respectively. 

The standard notion of General Position is $L(n,n,N;A)<1$, and Kruskal rank \cite{kruskal} is the largest $k$ such that $L(k,n,N;A)<1$.

Many of the theorems in compressed sensing rely upon a``sensing matrix" having suitable bounds on its RIC. Unfortunately, computing the RICs of a matrix $A$ is in general NP-hard, \cite{NPhard}. Efforts are underway to design algorithms which compute accurate bounds on the RICs of a matrix, \cite{daspromont,nemirovski} but to date these algorithms have a limited success, with the bounds only effective for $k \sim n^{1/2}.$ Lacking the ability to efficiently calculate the RICs of a given matrix, efforts are underway to compute probabilistic bounds for various random matrix ensembles. These efforts have followed three research programs:
\begin{itemize}
\item 
Determination of the largest ensemble of matrices such that as the problem sizes $(k,n,N)$ grow, the RICs $U(k,n,N;A)$ remains bounded and the $L(k,n,N;A)$ bounded away from 1 \cite{GAFA}.

\item
Computing as accurate bounds as possible for particular ensembles, such as the Gaussian ensemble \cite{Cand,BCT}, where the entries of $A$ are drawn i.i.d. from the standard Gaussian 
Normal ${\cal N}(0,1/n)$. (In part as a model for i.i.d. mean zero ensembles.)
 
\item
Computing as accurate bounds as possible for the partial Fourier ensemble \cite{Holg}, 
where $A$ is formed from random rows, $j$, or samples, $t_l$, of a Fourier matrix with 
entries $F_{j,l}=e^{2\pi i j t_l}$.  (In part as a model for matrices possessing a fast matrix vector product.) 

\end{itemize}

This manuscript focuses on the second of these research programs, accurate bounds for the Gaussian/Wishart ensemble. Cand\`es and Tao derived the first set of RIC bounds for the Gaussian ensemble using a union bound over all $N \choose k$ submatrices and bounding the singular values of each submatrix using concentration of measure bounds \cite{CTDecoding}. Blanchard, Cartis and Tanner derived the second set of RIC bounds for the Gaussian ensemble, similarly using a union bound over all $N \choose k$ submatrices, but achieved substantial improvements by using more accurate bounds on the probability density function of Wishart matrices \cite{BCT}. These bounds are presented here in Theorem \ref{CTthm} and Theorem \ref{BCTthm} respectively. This manuscript presents yet further improved bounds for the Gaussian ensemble, see Theorem \ref{BTthm} and Figure \ref{UBT_LBT}, by exploiting dependencies in the singular values of submatrices with overlapping support sets, say, $A_K$ and $A_{K'}$ with $|K\cap{K'}| \gg 1$.  These are the first RIC bounds that exploits this structure. In addition to asymptotic bounds for large problem sizes, we present bounds valid for finite values of $(k,n,N)$.

The manuscript is organised as follows: Our improved asymptotic bounds are stated in Section \ref{sec:ImpRIPbound} and their derivation described in Section \ref{sec:boundcostruct}. Prior bounds are presented in Section \ref{sec:PrRIPbound} and are compared with those in Theorem \ref{BTthm}.  Bounds valid for finite values are presented in Section \ref{sec:FiniteN}. A brief discussion on sparse approximation and compressed sensing and the implications of these bounds for compressed sensing  is given in Section \ref{sec:SACS}. Proof of technical lemmas used or assumed in our discussion come in the Appendix. 


\section{RIC Bounds}\label{sec:RIPbound} 
We focus our attention on bounding the RIC for the Gaussian ensemble in the setting of {\em proportional-growth asymptotics}.
\begin{definition}[Proportional-Growth Asymptotics]  
A sequence of problem sizes $(k,n,N)$ is said to follow proportional-growth asymptotics if, 
\begin{equation}
\label{lgasymp}
\frac{k}{n} = \rho_n \rightarrow \rho \quad \mathrm{and} \quad \frac{n}{N} = \delta_n \rightarrow \delta \quad \mathrm{for} \quad (\delta,\rho)\in (0,1)^2 \quad \mathrm{as} \quad (k,n,N) \rightarrow \infty.
\end{equation}
\end{definition}
In this asymptotic we provide quantitative values, above which it is exponentially unlikely 
that the RIC will exceed.  In Section \ref{sec:FiniteN} we show how our derivation of these 
bounds can also supply probabilities for specified bounds and finite values of $(k,n,N)$.

\subsection{Improved RIP Bounds}\label{sec:ImpRIPbound} 
The probability density functions (pdf) of the RIC for the Gaussian ensemble is currently unknown, but asymptotic probabilistic bounds have been proven. Our bounds, and earlier ones, for the RIC of the Gaussian ensemble built upon the bounds of the pdf's of the extreme eigenvalues of Gaussian (Wishart) matrices due to Edelman \cite{EdRao,edelman}.  All earlier bounds on the RIC have been derived using union bounds that consider each of the ${N \choose k}$ submatrices of size $n \times k$ individually \cite{BCT,CTDecoding}.  We consider groups of submatrices where the columns of the submatrices in a group are from at most $m\ge k$ distinct columns of $A$.  We present our improved bounds in Theorem \ref{BTthm}, preceded by the definition of the terms used in it given in Definition \ref{BTdfn}.  Plots of these bounds are displayed in Figure \ref{UBT_LBT}.

\begin{definition}
\label{BTdfn}
Let $(\delta,\rho)\in(0,1)^2$, $\gamma\in [\rho,\delta^{-1}]$, and denote the Shannon Entropy with base $e$ logarithms as $H(p):=p\ln(1/p)+(1-p)\ln(1/(1-p))$.  Let
\begin{align}
\label{psiminnew}
\psi_{min}\left(\lambda,\gamma\right) & := H\left(\gamma\right) + \frac{1}{2} \Big[ \left(1-\gamma \right)\ln\lambda + \gamma \ln\gamma + 1 - \gamma - \lambda \Big], \\
\label{psimaxnew}
\psi_{max}\left(\lambda,\gamma\right) & := \frac{1}{2} \Big[\left(1+\gamma\right)\ln\lambda - \gamma \ln\gamma + 1 + \gamma - \lambda \Big].
\end{align}

\noindent Define $\lambda^{min}(\delta,\rho;\gamma)$ and $\lambda^{max}(\delta,\rho;\gamma)$ as the solution to (\ref{lminnew}) and (\ref{lmaxnew}) respectively:
\begin{eqnarray}
\label{lminnew}
\delta \psi_{min}\left(\lambda^{min}(\delta,\rho;\gamma),\gamma\right) + H(\rho\delta) - \delta\gamma H\left(\rho/\gamma\right) = 0, 
\quad for ~\lambda^{min}(\delta,\rho;\gamma) \leq 1-\gamma, \\
\label{lmaxnew}
\delta \psi_{max}\left(\lambda^{max}(\delta,\rho;\gamma),\gamma\right) + H(\rho\delta) - \delta\gamma H\left(\rho/\gamma\right) = 0, \quad for ~\lambda^{max}(\delta,\rho;\gamma) \geq 1+\gamma.
\end{eqnarray}

\noindent Let $\displaystyle \lambda^{min}(\delta,\rho) := \mathop{\max}_{\gamma} \lambda^{min}(\delta,\rho;\gamma) \quad and \quad \lambda^{max}(\delta,\rho) := \mathop{\min}_{\gamma} \lambda^{max}(\delta,\rho;\gamma)$ and define
\begin{equation}
\Lbt(\delta,\rho) := 1 - \lambda^{min}(\delta,\rho) \quad and \quad \Ubt (\delta,\rho) := \lambda^{max}(\delta,\rho) - 1.
\end{equation}
\end{definition}

That for each $(\delta,\rho;\gamma)$, \eqref{lminnew} and \eqref{lmaxnew} have a unique solution $\lambda^{min}(\delta,\rho;\gamma)$ and $\lambda^{max}(\delta,\rho;\gamma)$ respectively was proven in \cite{BCT}.  That $\lambda^{min}(\delta,\rho;\gamma)$ and $\lambda^{max}(\delta,\rho;\gamma)$ have unique maxima and minima respectively over $\gamma\in [\rho,\delta^{-1}]$ is established in Lemma \ref{uniqminmax}.

\begin{theorem}
\label{BTthm}
Let $A$ be a matrix of size $n\times N$ whose entries are drawn i.i.d. from $\mathcal{N}(0,1/n)$.  Let $\delta$ and $\rho$ be defined as in \eqref{lgasymp}, and $\Lbt(\delta,\rho)$ and $\Ubt(\delta,\rho)$ be defined as in Definition \ref{BTdfn}.
For any fixed $\epsilon>0$,  in the proportional-growth asymptotic,
\[
\mathbf{P}(L(k,n,N) < \Lbt(\delta,\rho) + \epsilon) \rightarrow 1 \quad and \quad \mathbf{P}(U(k,n,N) < \Ubt(\delta,\rho) + \epsilon) \rightarrow 1
\]
exponentially in n.
\end{theorem}

In the spirit of reproducible research, software and web forms that evaluate $\Lbt(\delta,\rho)$ and $\Ubt(\delta,\rho)$ are publicly available at \cite{ecos}.

\begin{figure}[h!]
\centering
\includegraphics[width=0.495\textwidth]{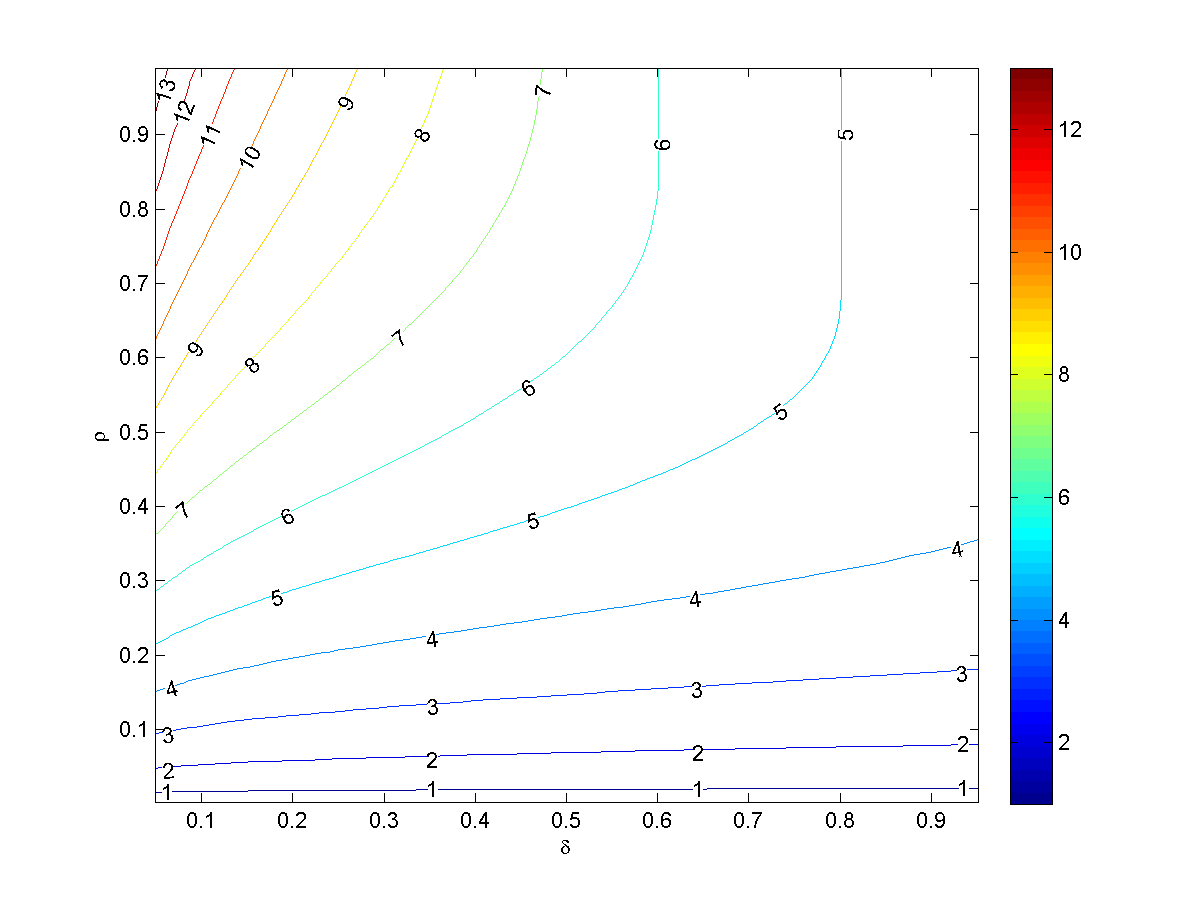}
\includegraphics[width=0.495\textwidth]{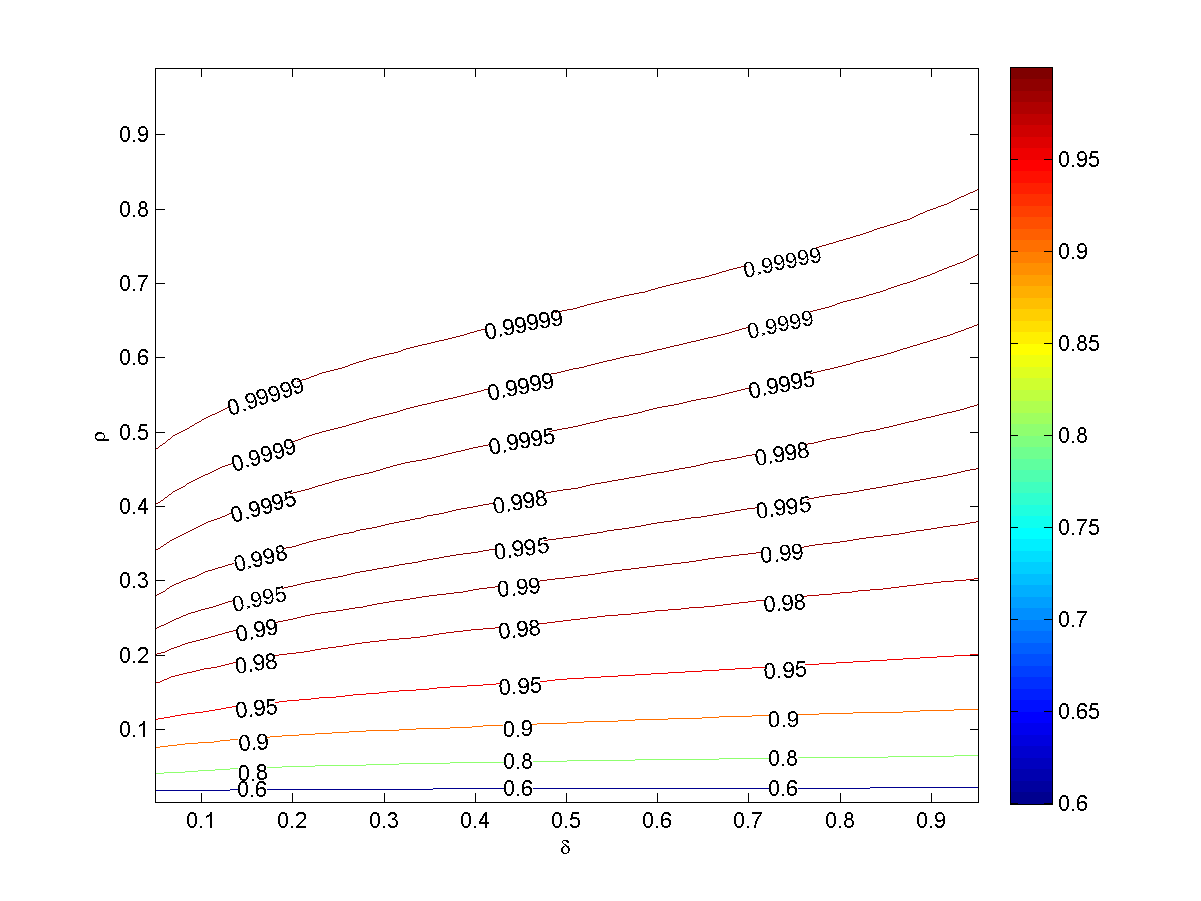}
\caption{$\Ubt(\delta,\rho)$ (left panel) and $\Lbt(\delta,\rho)$ (right panel) from Definition \ref{BTdfn} for $(\delta,\rho)\in (0,1)^2$.}
\label{UBT_LBT}
\end{figure}

Sharpness of the bounds can be probed by comparison with empirically observed lower bounds on the RIC for finite dimensional draws from the Gaussian ensemble.  There exist efficient algorithms for calculating lower bounds of RIC, \cite{Dossal,Ric}.  These algorithms perform local searches for submatrices with extremal eigenvalues.  The new bounds in Theorem \ref{BTthm}, see Figure \ref{UBT_LBT}, can be compared with empirical data displayed in Figure \ref{URic_LDos}.

\begin{figure}[h!]
\centering
\includegraphics[width=0.495\textwidth]{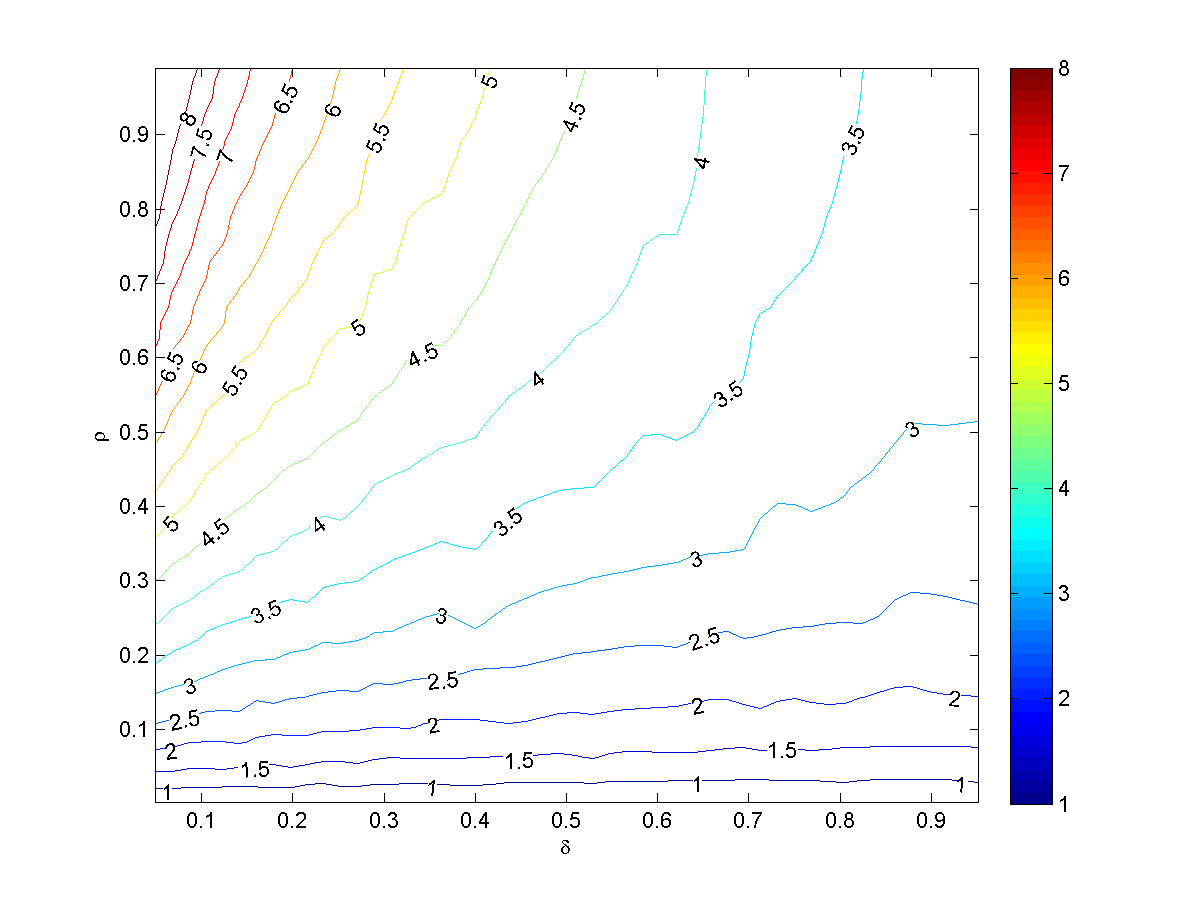}
\includegraphics[width=0.495\textwidth]{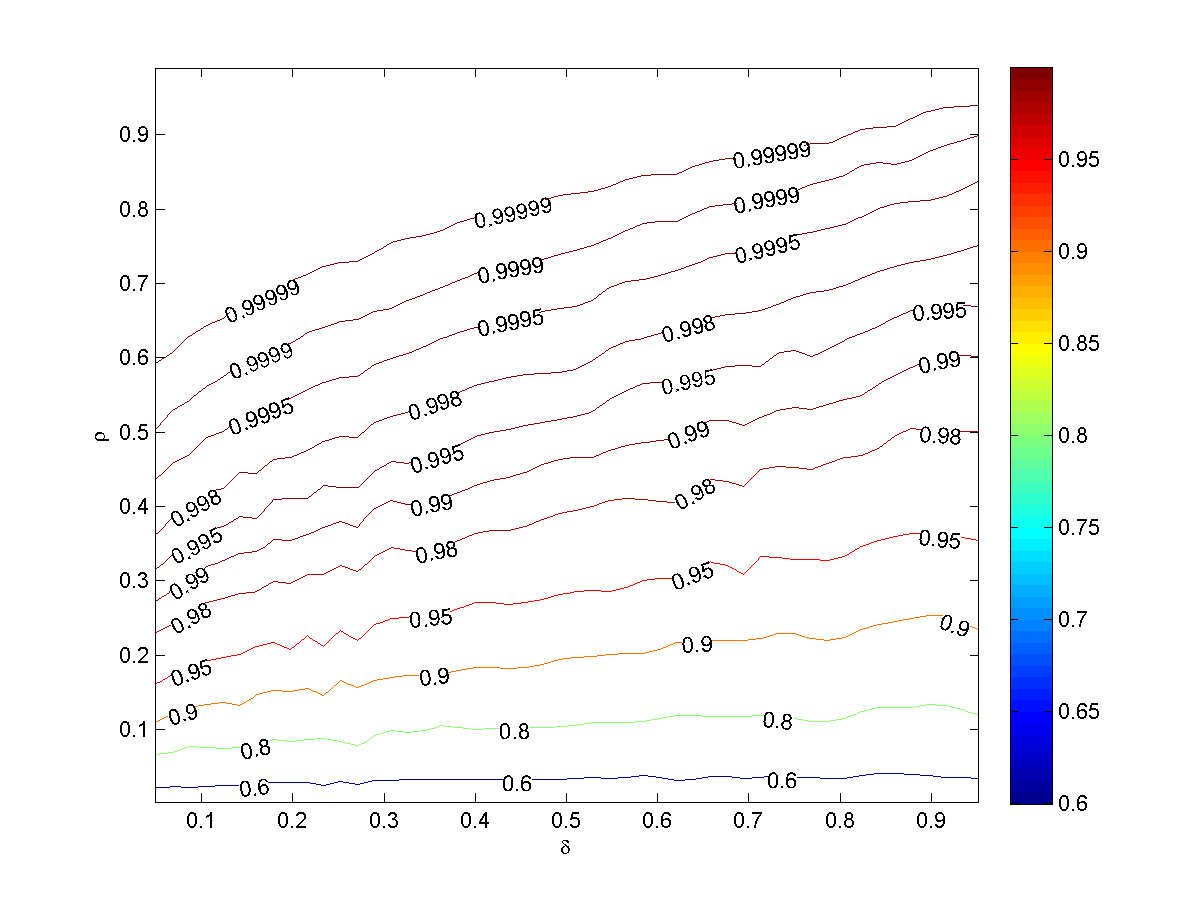}
\caption{\textbf{Empirically observed lower bounds on RIC for $A$ Gaussian.}
Observed lower bounds of $L(k, n,N;A)$ (left panel) and $U(k, n,N;A)$ (right panel). Although there is no computationally tractable method for calculating the RICs of a matrix, there are efficient algorithms which perform local searches for extremal eigenvalues of submatrices; allowing for observable lower bounds on the RICs. Algorithm for observing $L(k, n,N)$, \cite{Dossal}, and $U(k, n,N)$, \cite{Ric}, were applied to hundreds of $A$ drawn i.i.d. $\mathcal{N}(0,1/n)$ with $n = 400$ and $N$ increasing from $420$ to $8000$.}
\label{URic_LDos}
\end{figure}

To further demonstrate the sharpness of our bounds, we compute the maximum and minimum ``sharpness ratios'' of the bounds in Theorem \ref{BTthm} to empirically observed lower bounds; for each $\rho$, the maximum and minimum of the ratio is taken over all $\delta\in [0.05,0.9524]$, these are the same $\delta$ values used in Figure \ref{URic_LDos}. These ratios are shown in the left panel of Figure \ref{ULBT_ratios}, and are below 1.57 of the empirically observed lower bounds on $L(k,n,N)$ and $U(k,n,N)$ observed with $n=400$. 
\begin{figure}[h!]
\centering
\includegraphics[width=0.495\textwidth]{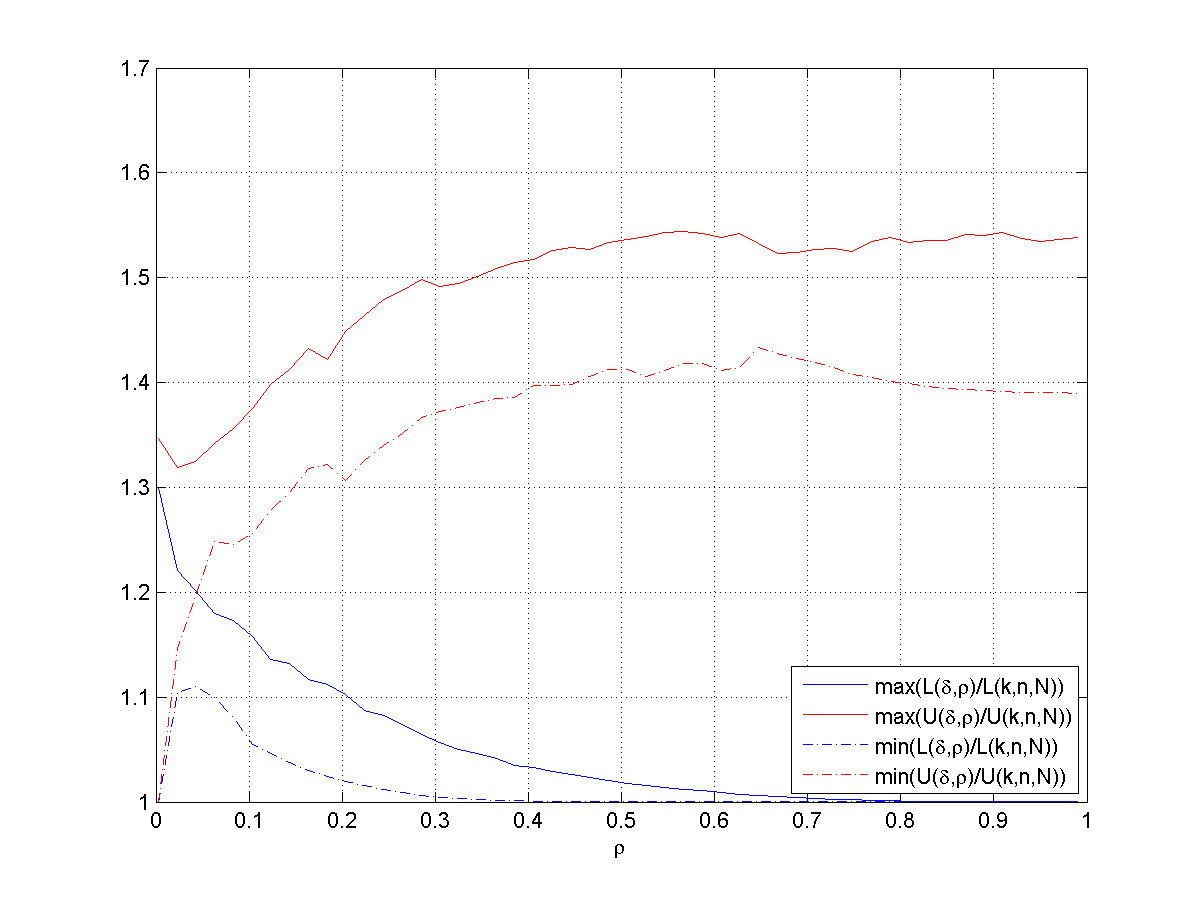}
\includegraphics[width=0.495\textwidth]{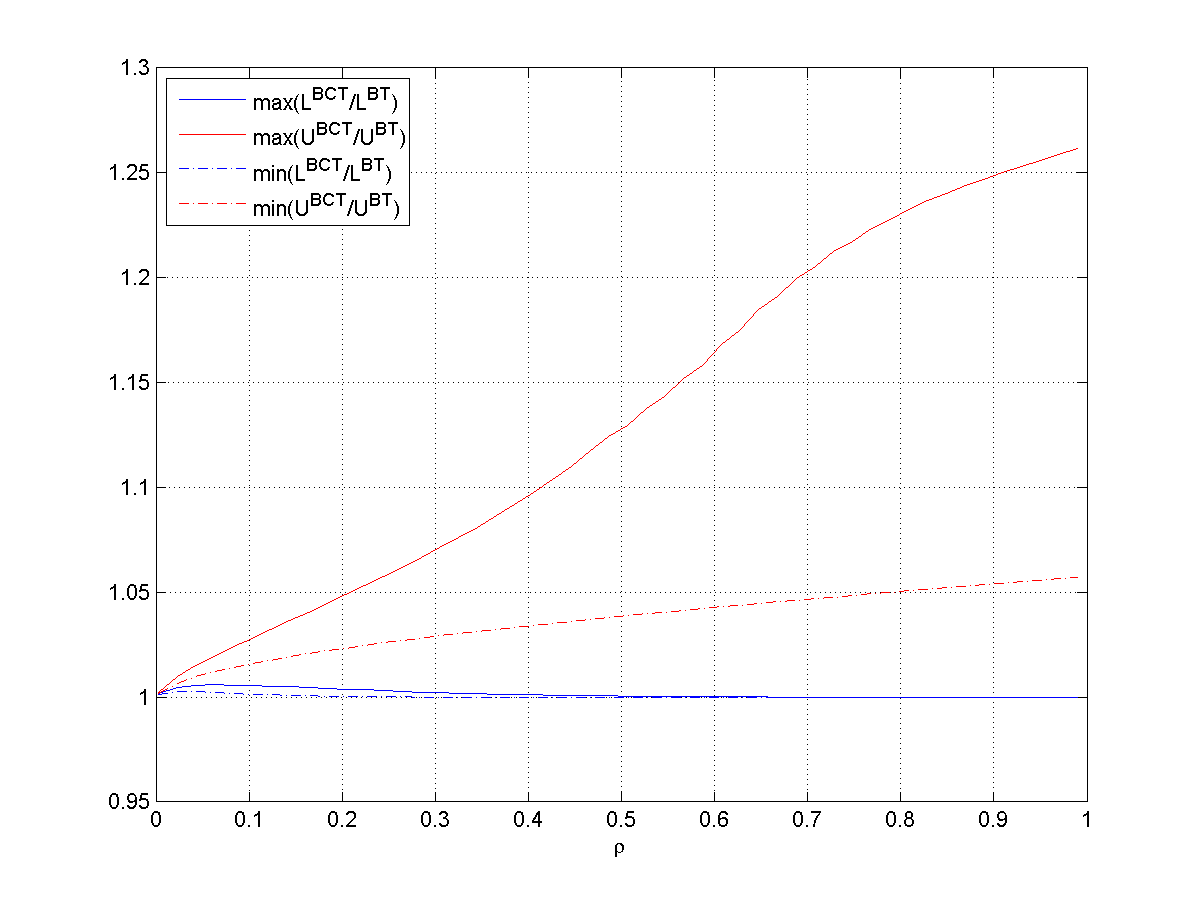}
\caption{Left panel: The maximum and minimum, over $\delta$, sharpness ratios, $\frac{\Ubt(\delta,\rho)}{U(k,n,N;A)}$ and $\frac{\Lbt(\delta,\rho)}{L(k,n,N;A)}$ as a function of $\rho$; with the maximum and minimum taken over all $\delta\in [0.05,0.9524]$, the same $\delta$ values used in Figure \ref{URic_LDos}.  Right panel:  The maximum and minimum, over $\delta$, improvement ratios over the previous best known bounds, $\frac{\Ubct(\delta,\rho)}{\Ubt(\delta,\rho)}$ and $\frac{\Lbct(\delta,\rho)}{\Lbt(\delta,\rho)}$ as a function of $\rho$; with the maximum and minimum also taken over $\delta\in [0.05,0.9524]$.}
\label{ULBT_ratios}
\end{figure}

\subsection{Discussion on the Construction of Improved RIC Bounds}\label{sec:boundcostruct} 
The bounds in Theorem \ref{BTthm} improve upon the earlier results of \cite{BCT} by grouping matrices $A_K$ and $A_{K'}$ which share a significant number of columns from $A$.  This is manifest in Definition \ref{BTdfn} through the introduction of the free parameter $\gamma$ which is associated with the size of group considered.  In this section we first discuss the way in which we construct these groups and the sense in which the bounds in Theorem \ref{BTthm} are optimal for this construction.  Equipped with a suitable construction of groups, we discuss the way in which this grouping is employed to improve the RIC bounds from \cite{BCT}.

\subsubsection{Construction of groups}
We construct our groups of $A_K$ by selecting a subset $M_i$ from ${1,2,\ldots,N}$ of cardinality $|M_i|=m\ge k$ and setting ${\cal G}_i:=\bigcup_{K\subset M_i,|K|=k} K$.  
The group ${\cal G}_i$ has ${m\choose k}$ members, with any two members sharing at least $2k-m$ elements.  Hence, the quantity $\gamma=m/n$ in Definition \ref{BTdfn} is associated with the cardinality of the groups ${\cal G}_i$.  In order to calculate bounds on the RIC of a matrix, we need a collection of groups whose union includes all ${N\choose k}$ sets of cardinality $k$ from $\Omega:=\{1,2,\ldots,N\}$; that is, we need $\{{\cal G}_i\}_{i=1}^u$ such that $G:=\bigcup_{i=1}^u {\cal G}_i$ with $|G|={N\choose k}$.  From simple counting, the minimum number of groups ${\cal G}_i$ needed for this covering is at least $r:={N \choose k}{m \choose k}^{-1}$.  Although the construction of a minimal covering is an open question \cite{Dan}, even a simple random construction of the ${\cal G}_i$'s requires typically only a polynomial multiple of $r$ groups, hence achieving the optimal large deviation rate.

\begin{lemma}[\cite{Dan}]\label{covering}
Set $r = {N \choose k} {m \choose k}^{-1}$ and draw $u:=r N$ sets $M_i$ each of cardinality $m$, drawn uniformly at random from the ${N \choose m}$ possible sets of cardinality $m$. With $G$ defined as above, 
\begin{equation}
\mathbf{P}\left(|G|<{N\choose k}\right)<C(k/N) N^{-1/2} e^{-N(1-\ln{2})}
\end{equation}
where $C(p)\le\frac{5}{4}(2\pi p(1-p))^{(-1/2)}$.
\end{lemma}

\begin{proof}
Select one set $K\subset \Omega$ of cardinality $|K|=k$ prior to the draw of the sets $M_i$.  The probability that it is not contained in one set $M_i$ is $1/r$, and with each $M_i$ drawn independently, the probability that it is not contained in any of the $u$ sets $M_i$ is 
$(1-r^{-1})^u\le e^{-u/r}$.  Applying a union bound over all ${N\choose k}$ sets $K$ yields 
\[
\mathbf{P}\left(|G|<{N\choose k}\right)<{N\choose k} e^{-u/r}.
\]
Noting from Stirling's Inequality that 
\begin{equation}\label{stirling}
\frac{16}{25}(2\pi p(1-p) N)^{(-1/2)} e^{NH(p)}\le
{N \choose pN}\le \frac{5}{4}(2\pi p(1-p) N)^{(-1/2)} e^{NH(p)},
\end{equation}
with $H(p)\le \ln{2}$ for $p\in[0,1]$,
and substituting in the selected value of $u$ completes the proof.  Note that an exponentially small probability can be obtained with $u$ just larger than $rNH(\delta\rho)$, but the smaller polynomial factor is negligible for our purposes.
\end{proof}

\begin{corollary}
\label{cover_corl}
Given Lemma \ref{covering}, as $n \rightarrow \infty$ in the proportional-growth asymptotics, the probability that all the ${N \choose k}$ $k$-subsets of ${1,2,\ldots,N}$ are covered by $G$ converges to one exponentially in $n$.
\end{corollary}

\subsubsection{Decreasing the combinatorial term}

We illustrate the way the groups ${\cal G}_i$ are used to improve the RIC bound on the upper RIC bound $U(k,n,N;A)$; the bounds for $L(k,n,N;A)$ following by a suitable replacement of maximizations/minimizations and sign changes.  All previous bounds on the RIC for the Gaussian ensemble have overcome the combinatorial maximization/minimization by use a union bound over all ${N \choose k}$ sets $K\subset \Omega$ and then using a tail bound on the pdf of the extreme eigenvalues of $A^*_KA_K$; for some $\lambda^*_{max}>0$,
\[
\mathbf{P}\left(\max_{K\subset\Omega,|K|=k} \lambda^{max}(A^*_KA_K)>\lambda_{max}^*\right) 
\le {N\choose k} \mathbf{P}\left(\lambda^{max}(A^*_KA_K)>\lambda_{max}^*\right)
\]
That the random variables $\lambda^{max}(A^*_KA_K)$ are treated as independent is the principal deficiency of this bound.  To exploit dependencies of this variable for $K$ and $K^{'}$ with significant overlap we exploit the groupings ${\cal G}_i$, which, at least for $m$ moderately larger than $k$, contain sets with significant overlap.  For the moment we assume the groups $\{{\cal G}_i\}_{i=1}^u$ cover all $K\subset \Omega$, and replace the above maximization over $K$ with a double maximizations
\[
\mathbf{P}\left(\max_{K\subset\Omega,|K|=k} \lambda^{max}(A^*_KA_K)>\lambda_{max}^*\right) 
=
\mathbf{P}\left(\max_{i=1,\ldots,u} \max_{K\subset{\cal G}_i,|K|=k} \lambda^{max}(A^*_KA_K)>\lambda_{max}^*\right).
\]
The outer maximization can be bounded over all $u$ sets ${\cal G}_i$, again, using a simple union bound; however, with a smaller combinatorial term.  The dependencies between $\lambda^{max}(A^*_KA_K)$ for $K\subset{\cal G}_i$ can be incorporated in the bound by replacing the maximization over $K\subset {\cal G}_i$ by $\lambda^{max}(A^*_{M_i} A_{M_i})$ where $M_i$ is the subset of cardinality $m$ containing all $K\subset {\cal G}_i$,
\begin{equation}\label{eq:unionbound}
\mathbf{P}\left(\max_{i=1,\ldots,u} \max_{K\subset{\cal G}_i,|K|=k} \lambda^{max}(A^*_KA_K)>\lambda_{max}^*\right) 
\le
u \mathbf{P}\left(\lambda^{max}(A^*_MA_M)>\lambda_{max}^* \right).
\end{equation}

Selecting $m=k$ recovers the usual union bound with $u$ equal to ${N\choose k}$.  Larger values of $m$ decrease the combinatorial term at the cost of increasing $\lambda^{max}(A^*_MA_M)$.  The efficacy of this approach depends on the interplay between these two competing factors.  In the proportional-growth asymptotic, this interplay is observed through the optimization over $\frac{m}{n}=\gamma\in [\rho,\delta^{-1}]$.  Definition \ref{BTdfn} uses the tail bounds on the extreme eigenvalues of Wishart Matrices derived by Edelman \cite{edelman} to bound $\mathbf{P}\left(\lambda^{max}(A^*_MA_M)>\lambda_{max}^* \right)$.  The previously best known bound on the RIC for the Gaussian ensemble is recovered by selecting $\gamma=\rho$ in Definition \ref{BTdfn}, \cite{BCT}.  The innovation of the bounds in Theorem \ref{BTthm} follows from there always being a unique $\gamma>\rho$ such that  $\lambda^{max}(\delta,\rho;\gamma)$ is less than $\lambda^{max}(\delta,\rho;\rho)$.

\begin{lemma}\label{uniqminmax}
Given that $\lambda^{min}(\delta,\rho;\gamma)$ and $\lambda^{max}(\delta,\rho;\gamma)$ are solutions to (\ref{lminnew}) and (\ref{lmaxnew}) respectively.  For any fixed $(\delta,\rho)$ there exist a unique $\gamma_{min}\in [\rho,\delta^{-1}]$ which minimizes $\lambda^{max}(\delta,\rho;\gamma)$ and a unique $\gamma_{max}\in [\rho,\delta^{-1}]$ which maximizes $\lambda^{min}(\delta,\rho;\gamma)$.  Furthermore, $\gamma_{min}$ and $\gamma_{max}$ are strictly larger than $\rho$.
\end{lemma}

\begin{figure}[h!]
\centering
\includegraphics[width=0.495\textwidth]{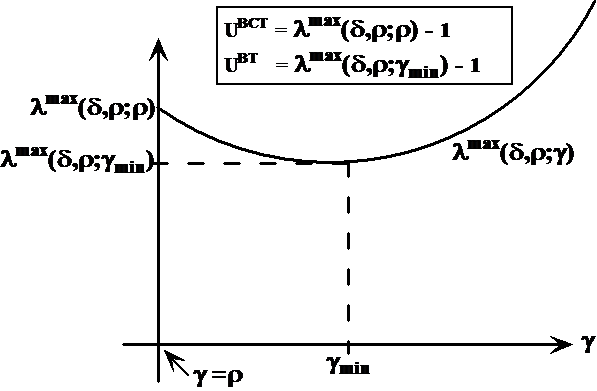}
\includegraphics[width=0.495\textwidth]{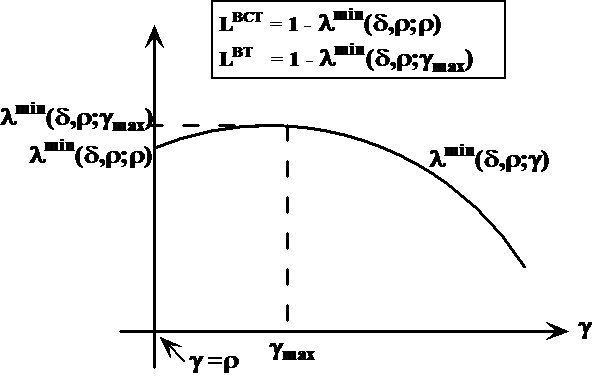}
\caption[$\lambda(\delta,\rho;\gamma)$]{Left panel: The relationship between the new bound $\Ubt(\delta,\rho)$, Theorem \ref{BTthm}, and the previous smallest bound $\Ubct(\delta,\rho)$, Theorem \ref{BCTthm}, where in the two bounds $\lambda^{max}(\delta,\rho;\gamma)$ is evaluated at $\gamma=\gamma_{min}$ and $\gamma=\rho$ respectively.  Right panel: The relationship between the new bound $\Lbt(\delta,\rho)$, Theorem \ref{BTthm}, and the previous smallest bound $\Lbct(\delta,\rho)$, Theorem \ref{BCTthm}, where in the two bounds $\lambda^{min}(\delta,\rho;\gamma)$ is evaluated at $\gamma=\gamma_{max}$ and $\gamma=\rho$ respectively.}.
\label{lambda}
\end{figure}

The optimal choices of $\gamma - \rho$ for $\Ubt(\rho,\delta)$ and $\Lbt(\rho,\delta)$ in $(\rho,\delta) \in (0,1)^{2}$ is displayed in Figure \ref{ULBT_gama}. 
The proof of Lemma \ref{uniqminmax} is presented in Appendix \ref{sec:Proofumm}.

\begin{figure}[h!]
\centering
\includegraphics[width=0.495\textwidth]{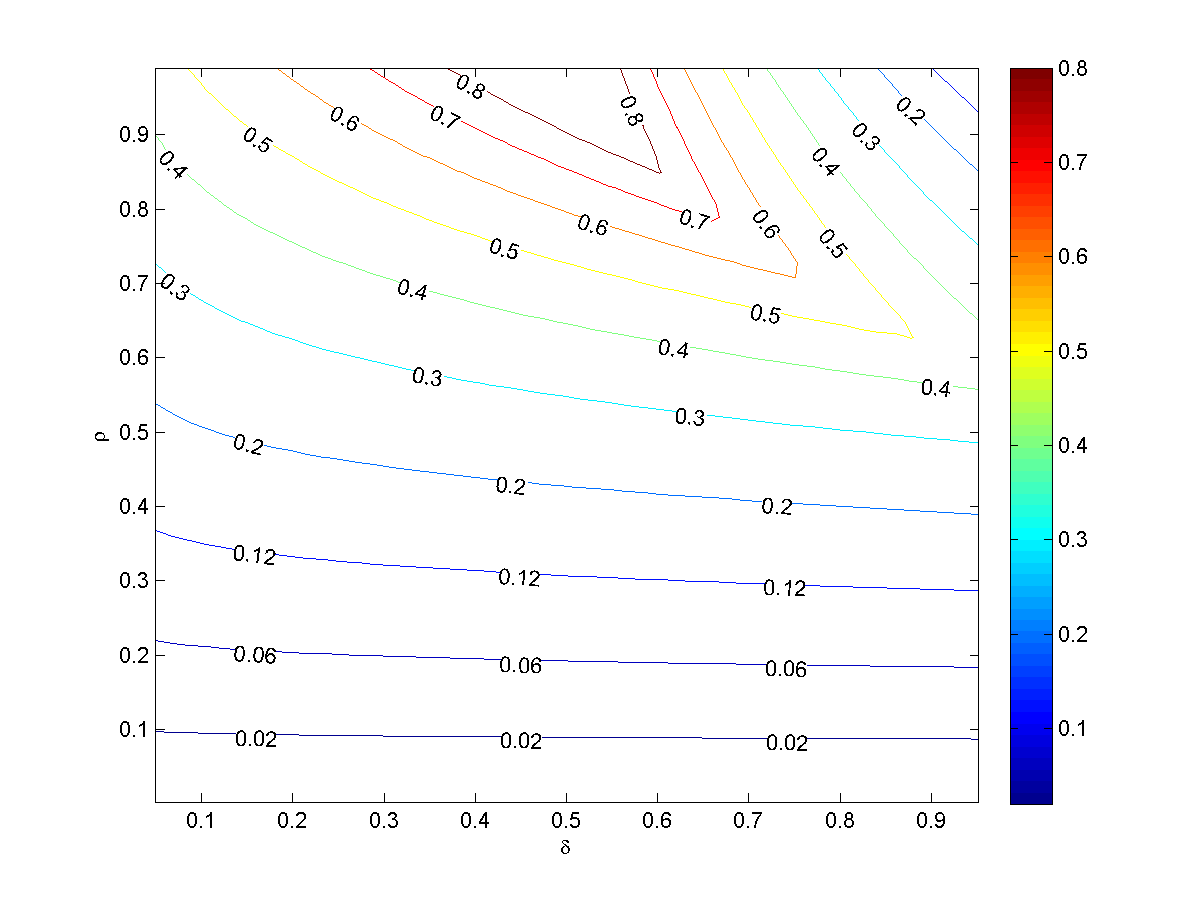}
\includegraphics[width=0.495\textwidth]{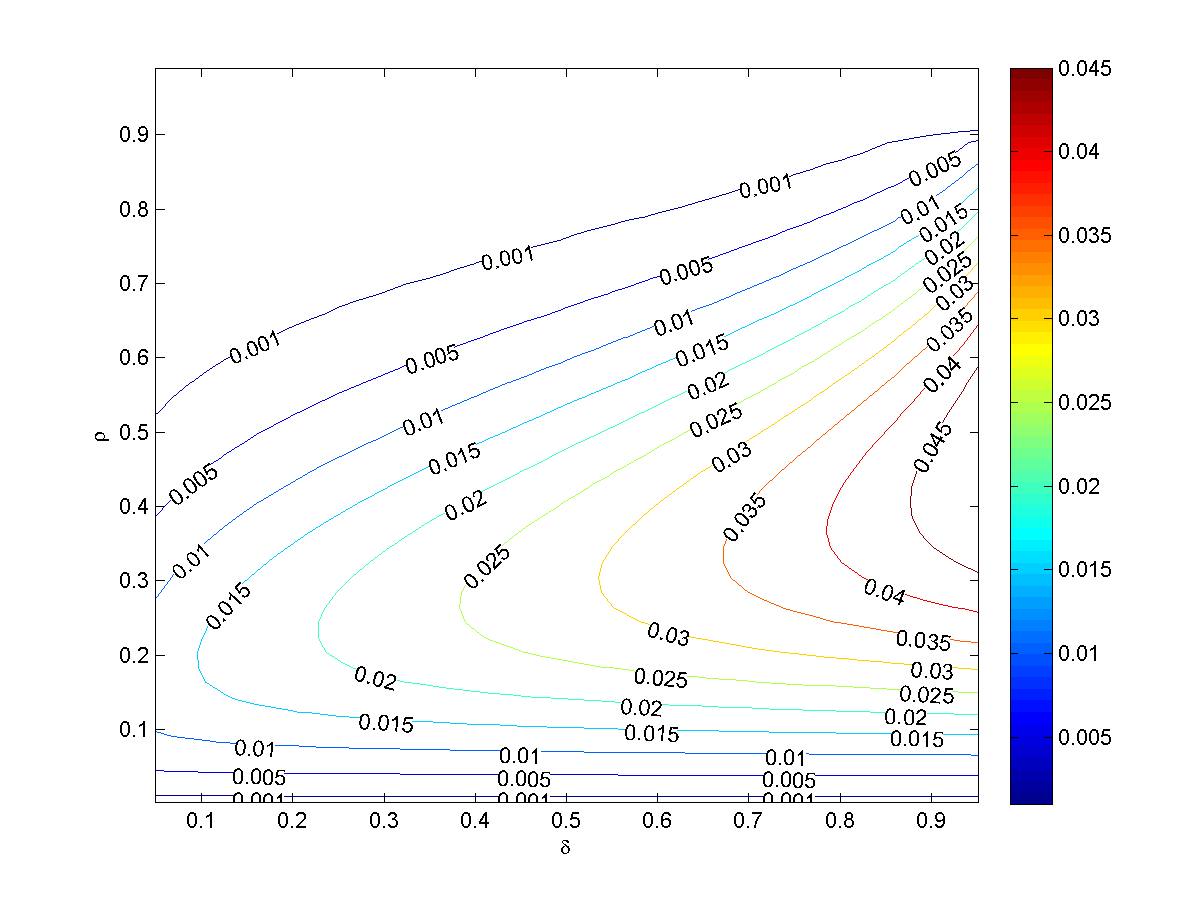}
\caption[$U(\delta,\rho;\gamma)$]{Optimal choice of $\gamma - \rho$ for $\Ubt(\delta,\rho)$ (left panel) and $\Lbt(\delta,\rho)$ (right panel).}
\label{ULBT_gama}
\end{figure}


\subsection{Prior RIP Bounds}\label{sec:PrRIPbound} 
There have been two previous quantitative bounds for the RIC of the Gaussian ensemble in the proportional-growth asymptotics.  The first bounds on the RIC of the Gaussian ensemble were supplied in \cite{CTDecoding} by Cand\`es and Tao using union bounds and concentration of measure bounds on the extreme eigenvalues of Wishart Matrices from \cite{ledoux}.  These bounds are stated in Theorem \ref{CTthm} with Definition \ref{CTdfn} defining some of the terms used in the theorem and plots of these bounds are displayed in Figure \ref{CTfig}.

\begin{definition}
\label{CTdfn}
Let $(\delta,\rho)\in (0,1)^2$ and define: 
\[
\Uct(\delta,\rho):=\left[1+\sqrt{\rho}+(2\delta^{-1}H(\delta\rho))^{1/2} \right]^2-1
\]
and
\[
\Lct(\delta,\rho):=  1-\max\left\{0,\left[ 1-\sqrt{\rho}-(2\delta^{-1}H(\delta\rho))^{1/2} \right]^2\right\}.
\]
\end{definition}

\begin{theorem}[Cand\'es and Tao \cite{CTDecoding}]
\label{CTthm}
Let $A$ be a matrix of size $n\times N$ whose entries are drawn i.i.d. from $\mathcal{N}(0,1/n)$.  Let $\delta$ and $\rho$ be defined as in \eqref{lgasymp}, and $\Lct(\delta,\rho)$ and $\Uct(\delta,\rho)$ be defined as in Definition \ref{CTdfn}.
For any fixed $\epsilon>0$,  in the proportional-growth asymptotic,
\[
\mathbf{P}(L(k,n,N) < \Lct(\delta,\rho) + \epsilon) \rightarrow 1 \quad and \quad \mathbf{P}(U(k,n,N) < \Uct(\delta,\rho) + \epsilon) \rightarrow 1
\]
exponentially in n.
\end{theorem}

The bounds in Theorem \ref{BTthm} follow the construction of the second bounds on the RIC for the Gaussian ensemble, presented in \cite{BCT}.  Removing the optimization of $\gamma\in[\rho,\delta^{-1}]$ in Definition \ref{BTdfn} and fixing $\gamma=\rho$ recovers the bounds on $L(k,n,N;A)$ and the first of two bounds on $U(k,n,N;A)$ presented in \cite{BCT}.  The first bound on $U(k,n,N;A)$ in \cite{BCT} suffer from excessive overestimation when $\delta\rho\approx 1/2$ due to the combinatorial term.  In fact, this overestimation is so severe that for some $(\delta,\rho)$ with $\delta\rho\approx 1/2$, smaller bounds are obtained at $(\delta,1)$.  This overestimation is somewhat ameliorated by noting the monotonicity of $U(k,n,N;A)$ in $k$, obtaining the improved bound, see \eqref{eq:Ubct}.  These bounds are stated in Theorem \ref{BCTthm} with Definition \ref{BCTdfn} defining some of the terms used in the theorem and plots of these bounds are displayed in Figure \ref{BCTfig}.

\begin{definition}
\label{BCTdfn} 
Let $(\delta,\rho)\in(0,1)^2$, and denote the Shannon Entropy with base $e$ logarithms as $H(p):=p\ln(1/p)+(1-p)\ln(1/(1-p))$.  Let $\psi_{min}(\lambda,\rho)$ and $\psi_{max}(\lambda,\rho)$ be defined as in \eqref{psiminnew} and \eqref{psimaxnew} respectively.
\noindent\textit{Define $\lambda^{min}(\delta,\rho)$ and $\lambda^{max}(\delta,\rho)$ as the solution to (\ref{lmin}) and (\ref{lmax}) respectively:} 
\begin{equation}
\label{lmin}
\delta \psi_{min}(\lambda^{min}(\delta,\rho),\rho) + H(\rho\delta) = 0 \quad for \quad \lambda^{min}(\delta,\rho) \leq 1-\rho,
\end{equation}
\begin{equation}
\label{lmax}
\delta \psi_{max}(\lambda^{max}(\delta,\rho),\rho) + H(\rho\delta) = 0 \quad for \quad \lambda^{max}(\delta,\rho) \geq 1+\rho.
\end{equation}

\noindent\textit{Define $\Lbct(\delta,\rho)$ and $\Ubct(\delta,\rho)$ as}
\begin{equation}\label{eq:Ubct}
\Lbct(\delta,\rho) := 1 - \lambda^{min}(\delta,\rho) \quad and \quad 
\Ubct(\delta,\rho) := \mathop{\min}_{\nu \in [\rho,1]} \lambda^{max}(\delta,\nu) - 1.
\end{equation}
\end{definition}

\begin{figure}[h!]
\centering
\includegraphics[width=0.495\textwidth]{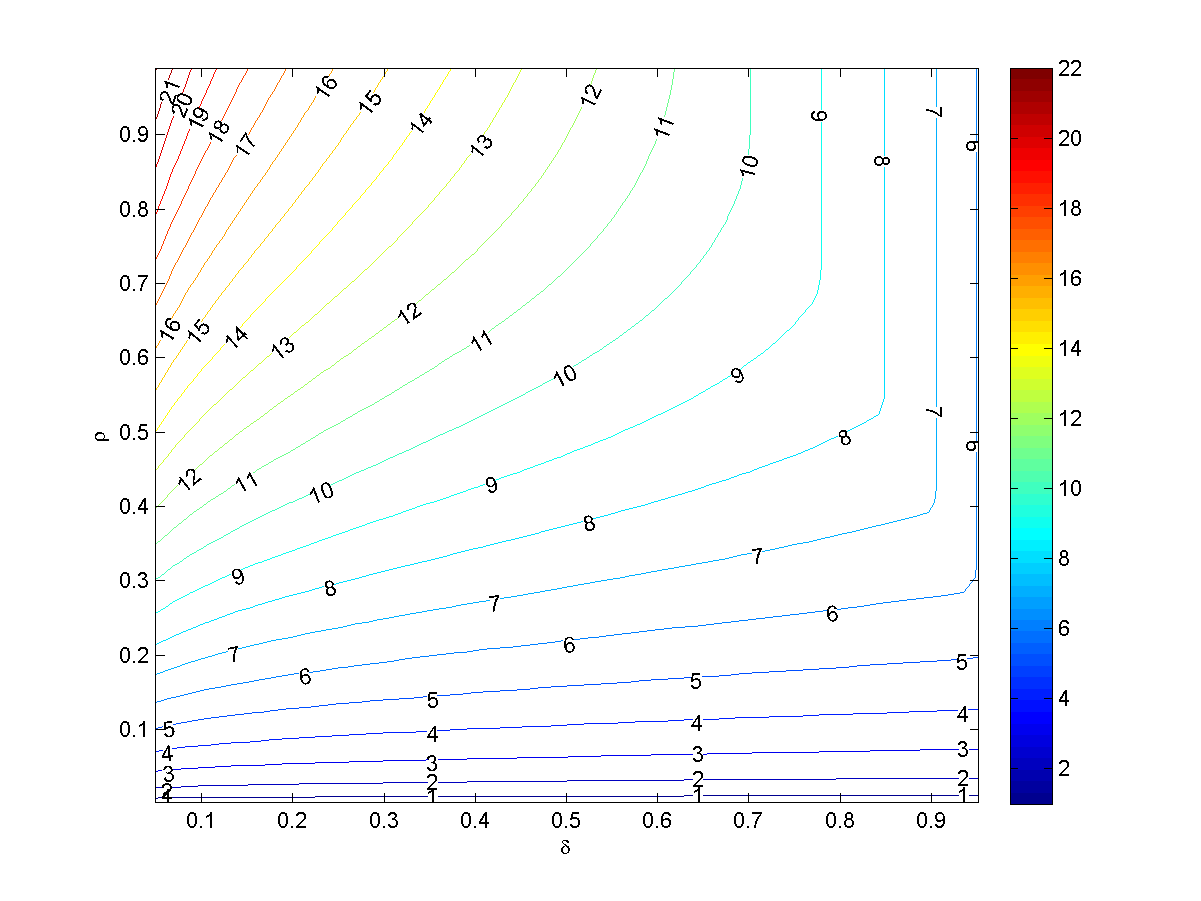}
\includegraphics[width=0.495\textwidth]{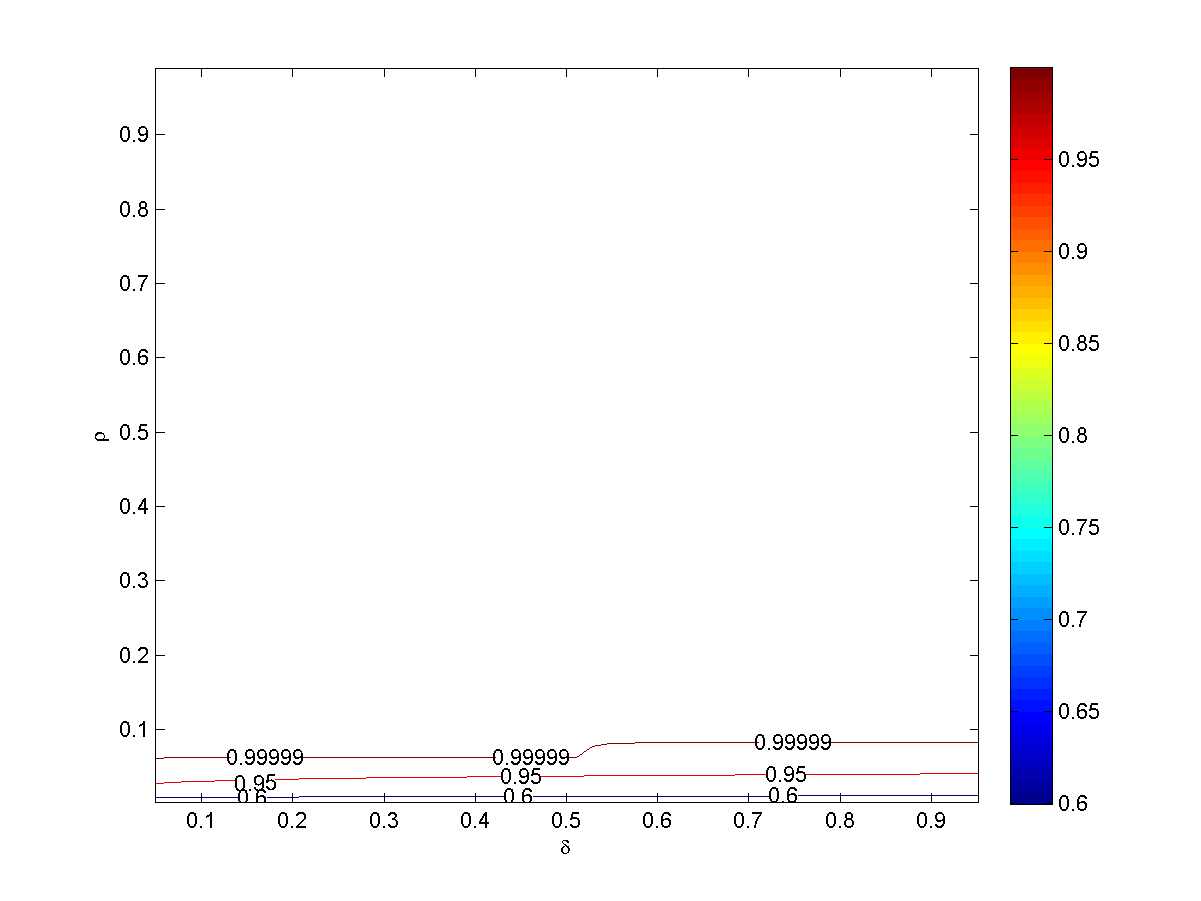}
\caption{$\Uct(\delta,\rho)$ (left panel) and  $\Lct(\delta,\rho)$ (right panel) from Definition \ref{CTdfn} for $(\delta,\rho)\in (0,1)^2$.}
\label{CTfig}
\end{figure}

\begin{figure}[h!]
\centering
\includegraphics[width=0.495\textwidth]{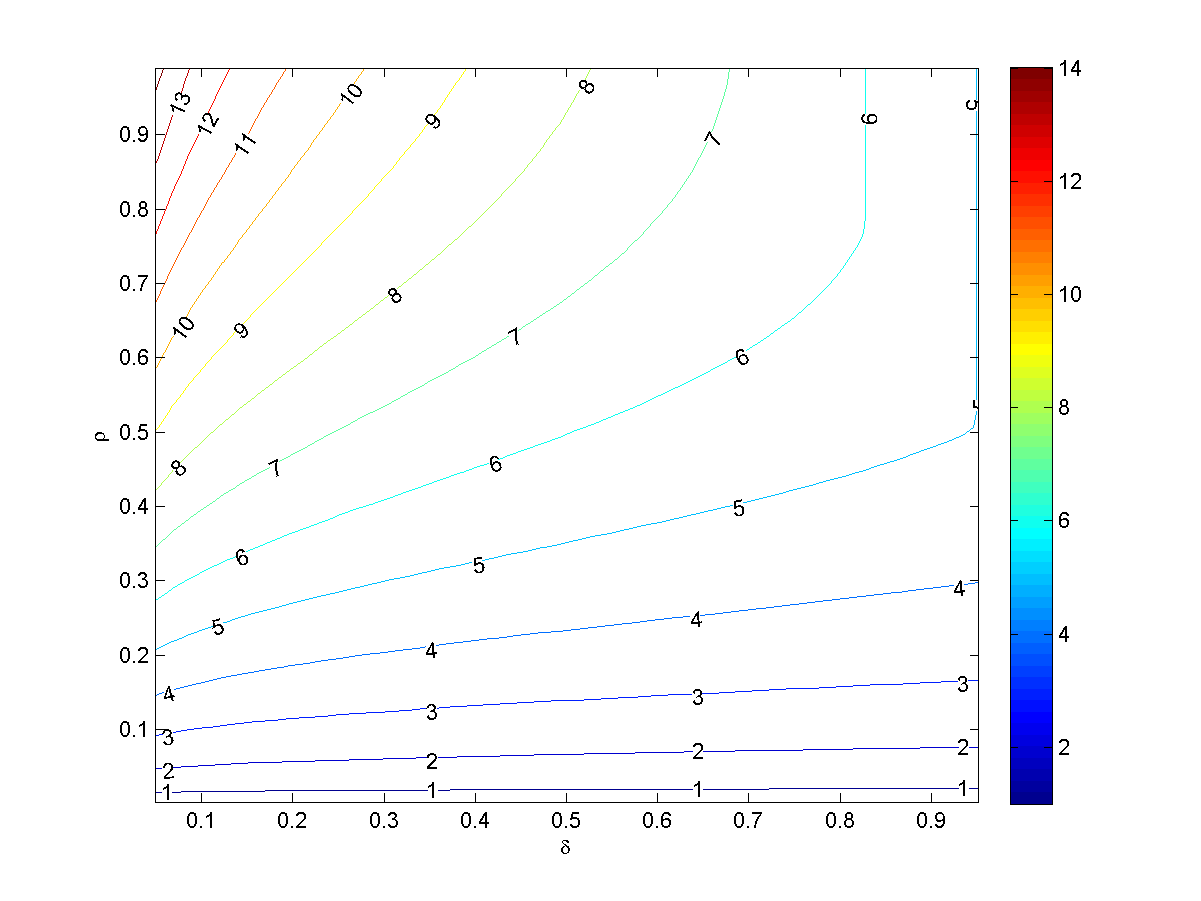}
\includegraphics[width=0.495\textwidth]{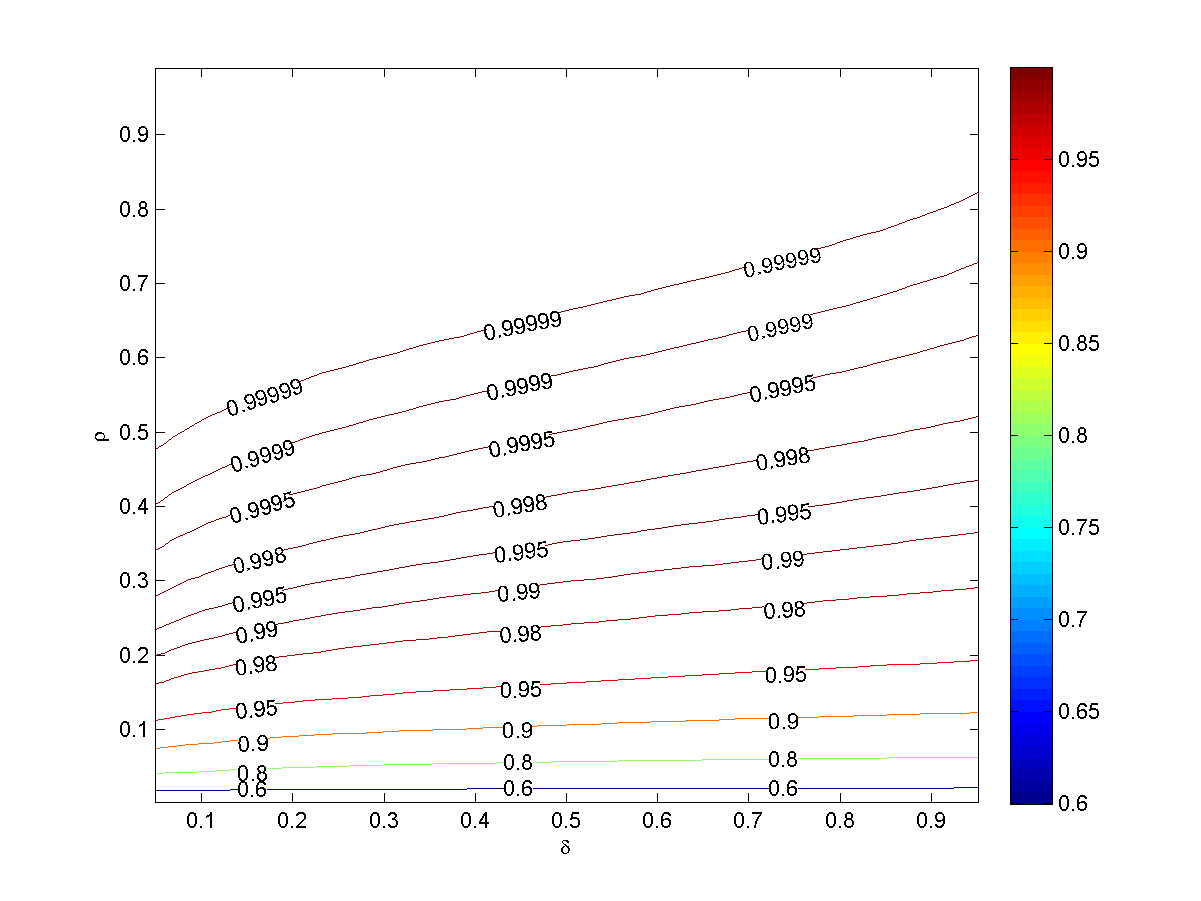}
\caption[$\Ubct(\delta,\rho)$]{$\Ubct(\delta,\rho)$ (left panel) and $\Lbct(\delta,\rho)$ (right panel) from Definition \ref{BCTdfn} for $(\delta,\rho)\in (0,1)^2$.}
\label{BCTfig}
\end{figure}

\begin{theorem}[Blanchard, Cartis, and Tanner \cite{BCT}]\label{BCTthm} 
Let $A$ be a matrix of size $n\times N$ whose entries are drawn i.i.d. from $\mathcal{N}(0,1/n)$.  Let $\delta$ and $\rho$ be defined as in \eqref{lgasymp}, and $\Lbct(\delta,\rho)$ and $\Ubct(\delta,\rho)$ be defined as in Definition \ref{BCTdfn}.
For any fixed $\epsilon>0$,  in the proportional-growth asymptotic,
\[
\mathbf{P}(L(k,n,N;A) < \Lbct(\delta,\rho) + \epsilon) \rightarrow 1 \quad and \quad \mathbf{P}(U(k,n,N;A) < \Ubct(\delta,\rho) + \epsilon) \rightarrow 1
\]
exponentially in n.
\end{theorem}

Figures \ref{CTfig} and \ref{BCTfig} show that the bounds in Theorem \ref{BCTthm} are a substantial improvement to those in Theorem \ref{CTthm}. The bounds presented here in Definition \ref{BTdfn} and Theorem \ref{BTthm} are a further improvement over those in \cite{BCT}, as implied by Lemma \ref{uniqminmax}.

\begin{corollary}\label{improvement}
Let $\Lbt(\delta,\rho)$ and $\Ubt(\delta,\rho)$ be defined as in Definition \ref{BTdfn} and $\Lbct(\delta,\rho)$ and $\Ubct(\delta,\rho)$ be defined as in Definition \ref{BCTdfn}.  For any fixed $(\delta,\rho)\in (0,1)^2$,
\[
\Lbt(\delta,\rho)<\Lbct(\delta,\rho)\quad\quad \Ubt(\delta,\rho)<\Ubct(\delta,\rho).
\]
\end{corollary}

The right panel of Figure \ref{ULBT_ratios} shows the ratio of the previously best known bounds, Definition \ref{BCTdfn} to the new bounds, Definition \ref{BTdfn}; for each $\rho$, the ratio is maximized over $\delta\in [0.05,0.9524]$.


\subsection{Finite $N$ Interpretations}\label{sec:FiniteN} The method of proof used to obtain the proportional-growth asymptotic bounds in Definition \ref{BTdfn}  also provides, albeit less elegant,  bounds valid for finite values of $(k,n,N)$ and specified probabilities of the bound being satisfied.  For a specified problem instance $(k,n,N)$ and $\epsilon$, bounds on the probabilities $\mathbf{P}\left(U(k,n,N) > \Ubt(\delta,\rho) + \epsilon \right)$ and $\mathbf{P}\left(L(k,n,N) > \Lbt(\delta,\rho) + \epsilon \right)$ are given in Propositions \ref{propU} and \ref{propL} respectively. \\

\begin{proposition} \label{propU} Let $A$ be a matrix of size $n\times N$ whose entries are drawn i.i.d. from $\mathcal{N}(0,1/n).$ Define $\Ubt(\delta,\rho)$ as in Definition \ref{BTdfn}. Then for any $\epsilon> 0,$
\begin{eqnarray}
\mathbf{P}\Big(U(k,n,N) > U(\delta_{n},\rho_{n}) + \epsilon \Big) & \leq & p'_{max}\Big(n,\lambda^{max}(\delta_{n},\rho_{n})\Big) \exp\Bigg(n\epsilon \cdot \frac{d}{d\lambda} \psi_{U}\bigg(\lambda^{max}(\delta_{n},\rho_{n})\bigg)\Bigg)  \nonumber\\
\label{ppmax}
& + & \frac{5}{4}(2\pi k(1-k/N))^{-1/2}\exp(-N(1-\ln{2})) ,
\end{eqnarray}
where  
\begin{equation}
p'_{max}(n,\lambda):= \left(\frac{8}{\pi}\right)^{1/2}\frac{2 n^{-7/2}}{\sqrt{\gamma\lambda}} 
\left(\frac{5}{4}\right)^3\left(\frac{nN(\gamma-\rho)}{\gamma\delta(1-\rho\delta)}\right)^{1/2},
\end{equation}
and
\[ \psi_U(\lambda,\gamma):=\delta^{-1}\left[H(\rho\delta) - \delta\gamma H\left(\frac{\rho}{\gamma}\right) + \delta \psi_{max}(\lambda, \gamma) \right] \]
for $\psi_{max}(\lambda, \gamma)$ define in (\ref{psimaxnew}). \\
\end{proposition}

\begin{proposition} \label{propL} Let $A$ be a matrix of size $n\times N$ whose entries are drawn i.i.d. from $\mathcal{N}(0,1/n).$ Define $\Lbt(\delta,\rho)$ as in Definition \ref{BTdfn}.  Then for any $\epsilon> 0,$ 
\begin{eqnarray}
\mathbf{P}\Big(L(k,n,N) > L(\delta_{n},\rho_{n}) + \epsilon \Big) & \leq & p'_{min}\Big(n,\lambda^{min}(\delta_{n},\rho_{n})\Big) \exp\Bigg(n\epsilon \cdot \frac{d}{d\lambda} \psi_{L}\bigg(\lambda^{min}(\delta_{n},\rho_{n})\bigg)\Bigg)\nonumber\\
\label{ppmin}
& + & \frac{5}{4}(2\pi k(1-k/N))^{-1/2}\exp(-N(1-\ln{2})),
\end{eqnarray}
where
\begin{equation}
p'_{min}(n,\lambda):=\left(\frac{5}{4}\right)^3  \frac{e\sqrt{\lambda}}{\pi\sqrt{2}} \left(\frac{nN(\gamma-\rho)}{\gamma\delta(1-\rho\delta)}\right)^{1/2},
\end{equation}
and
\[ \psi_L(\lambda,\gamma):=\delta^{-1}\left[H(\rho\delta) - \delta\gamma H\left(\frac{\rho}{\gamma}\right) + \delta \psi_{min}(\lambda, \gamma) \right] \]
for $\psi_{min}(\lambda, \gamma)$ define in (\ref{psiminnew}). \\
\end{proposition}

The proofs of Propositions \ref{propU} and \ref{propL} are presented in Appendix \ref{sec:ProofRIP} and also serve as the proof of Theorem \ref{BTthm} which follows by taking the appropriate limits.  From Propositions \ref{propU} and \ref{propL} we calculated bounds for a few example values of $(k,n,N)$ and $\epsilon$. Table \ref{tab:FiniteN1} shows bounds on $\mathbf{P}\left(U(k,n,N) > \Ubt(\delta,\rho) + \epsilon \right)$ for a few values of $(k,n,N)$ with two different choices of $\epsilon$. It is remarkable that these probabilities are already close to zero for these small values of $(k,n,N)$ and even for $\epsilon\ll 1$.  Table \ref{tab:FiniteN2} shows bounds on $\mathbf{P}\left(L(k,n,N) > \Lbt(\delta,\rho) + \epsilon \right)$ for the same values of $(k,n,N)$ as in Table \ref{tab:FiniteN1}, but with even smaller values for $\epsilon$.  Again, it is remarkable that these probabilities are extremely small, even for relatively small values of $(k,n,N)$ and $\epsilon$.

\begin{table}[h]
\centering
\begin{tabular}{|c|c|c|c|c|}
\hline
$k$ & $n$ & $N$ & $\epsilon$ & $Prob$ \\
\hline
100 & 200 & 2000 & $10^{-3}$ & $2.9 \times 10^{-2}$ \\
200 & 400 & 4000 & $10^{-3}$ & $9.5 \times 10^{-3}$ \\
400 & 800 & 8000 & $10^{-3}$ & $2.9 \times 10^{-3}$ \\
\hline
100 & 200 & 2000 & $10^{-10}$ & $3.2 \times 10^{-2}$ \\
200 & 400 & 4000 & $10^{-10}$ & $1.1 \times 10^{-2}$ \\
400 & 800 & 8000 & $10^{-10}$ & $4.0 \times 10^{-3}$ \\
\hline
\end{tabular}
\caption{$Prob$ is an upper bound of $\mathbf{P}\left(U(k,n,N) > \Ubt(\delta,\rho) + \epsilon \right)$ for the specified $(k,n,N)$ and $\epsilon$.}
\label{tab:FiniteN1}
\end{table}

\begin{table}[h]
\centering
\begin{tabular}{|c|c|c|c|c|}
\hline
$k$ & $n$ & $N$ & $\epsilon$ & $Prob$ \\
\hline
100 & 200 & 2000 & $10^{-5}$ & $2.8 \times 10^{-18}$ \\
200 & 400 & 4000 & $10^{-5}$ & $9.1 \times 10^{-32}$ \\
400 & 800 & 8000 & $10^{-5}$ & $2.8 \times 10^{-58}$ \\
\hline
\end{tabular}
\caption{$Prob$ is an upper bound of $\mathbf{P}\left(L(k,n,N) > \Lbt(\delta,\rho) + \epsilon \right)$ for the specified $(k,n,N)$ and $\epsilon$.}
\label{tab:FiniteN2}
\end{table}


\subsection{Implications for Sparse Approximation and Compressed Sensing}\label{sec:SACS}.

The RIC were introduced by Cand\'es and Tao \cite{CTDecoding} as a technique to prove that in certain conditions the sparsest solution of an underdetermined system of equations $Ax=b$ ($A$ of size $n\times N$ with $n<N$) can be found using linear programming.  The RIC is now a widely used technique in the study of sparse approximation algorithms, allowing the analysis of sparse approximation algorithms without specifying the measurement matrix $A$.  For instance, in \cite{candes_sqrt2} it was proven that if $\max(L(k,n,N;A),U(k,n,N;A))<\sqrt{2}-1$ then if $Ax=b$ has a unique $k$-sparse solution as its sparsest solution, then $\argmin\|z\|_1$ subject to $Az=b$ will be this $k$-sparse solution.  A host of other RIC based conditions have been derived for this and other sparsifying algorithms.  However, the values of $(k,n,N)$ when these conditions on the RIC are satisfied can only be determined once the measurement matrix $A$ has been specified \cite{BCT_spl}.  

The RIC bounds for the Gaussian ensemble discussed here allow one to state values of $(k,n,N)$ when sparse approximation recovery conditions are satisfied; and from these, guarantee the recovery of $k$-sparse vectors from $(A,b)$.  Unfortunately, all existing sparse approximation bound on the RICs are sufficiently small that they are only satisfied for $\rho\ll 1$, typically on the order of $10^{-3}$.  Although the bounds presented here are a strict improvement over the previously best known bounds, and for some $(\delta,\rho)$ achieve as much as a $20\%$ decrease, see Figure \ref{ULBT_ratios}, the improvements for $\rho\ll 1$ are meager, approximately $0.5-1\%$.  This limited improvement for compressed sensing algorithms is in large part due to the previous bounds being within $30\%$ of empirically observed lower bounds on RIC for $n=400$ when $\rho<10^{-2}$, \cite{BCT}.  

In \cite{BCTT}, using RIC bounds from \cite{BCT}, lower bounds on the phase transitions for exact recovery of \mbox{$k-\mathrm{sparse}$} signals for three greedy algorithms and \mbox{$l_{1}-\mathrm{minimisation}$} were presented. These curves are functions $\rho_{S}^{sp}(\delta)$ for Subspace Pursuit (SP) \cite{SubspacePursuit}, $\rho_{S}^{csp}(\delta)$ for Compressive Sampling Matching Pursuit (CoSaMP) \cite{NeTr09_cosamp}, $\rho_{S}^{iht}(\delta)$ for Iterative Hard Thresholding (IHT) \cite{BlDa08_iht} and $\rho_{S}^{l_{1}}(\delta)$ for \mbox{$l_{1}-\mathrm{minimisation}$} \cite{FoucartLai08}.  Figure 1 in \cite{BCTT} shows a plot of these phase transition curves. Figure \ref{phaseT} shows the new phase transition curves based on our new bounds.  The curves in the left Panel are approximately $0.5-1\%$ higher than those presented in \cite{BCT}.

\begin{figure}[ht]
\centering
\includegraphics[width=0.495\textwidth]{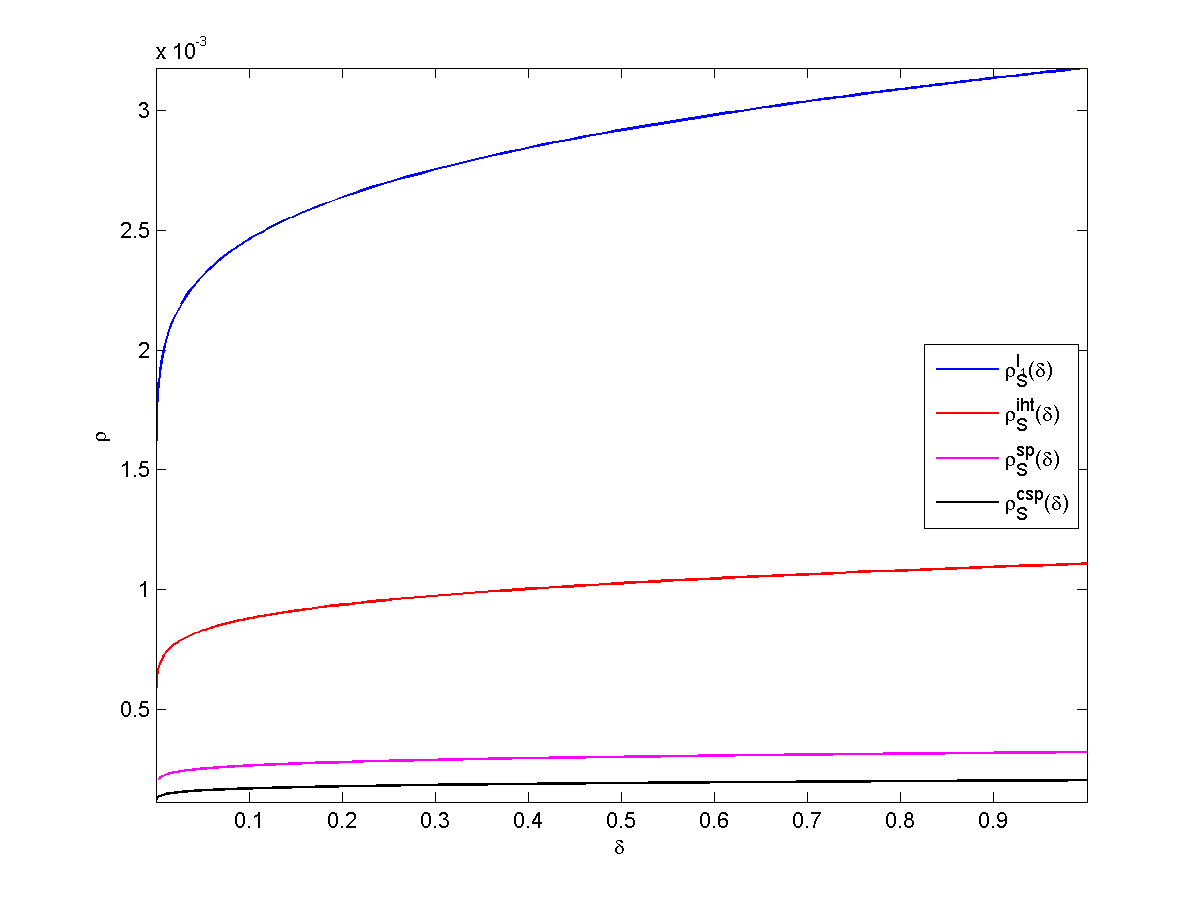}
\includegraphics[width=0.495\textwidth]{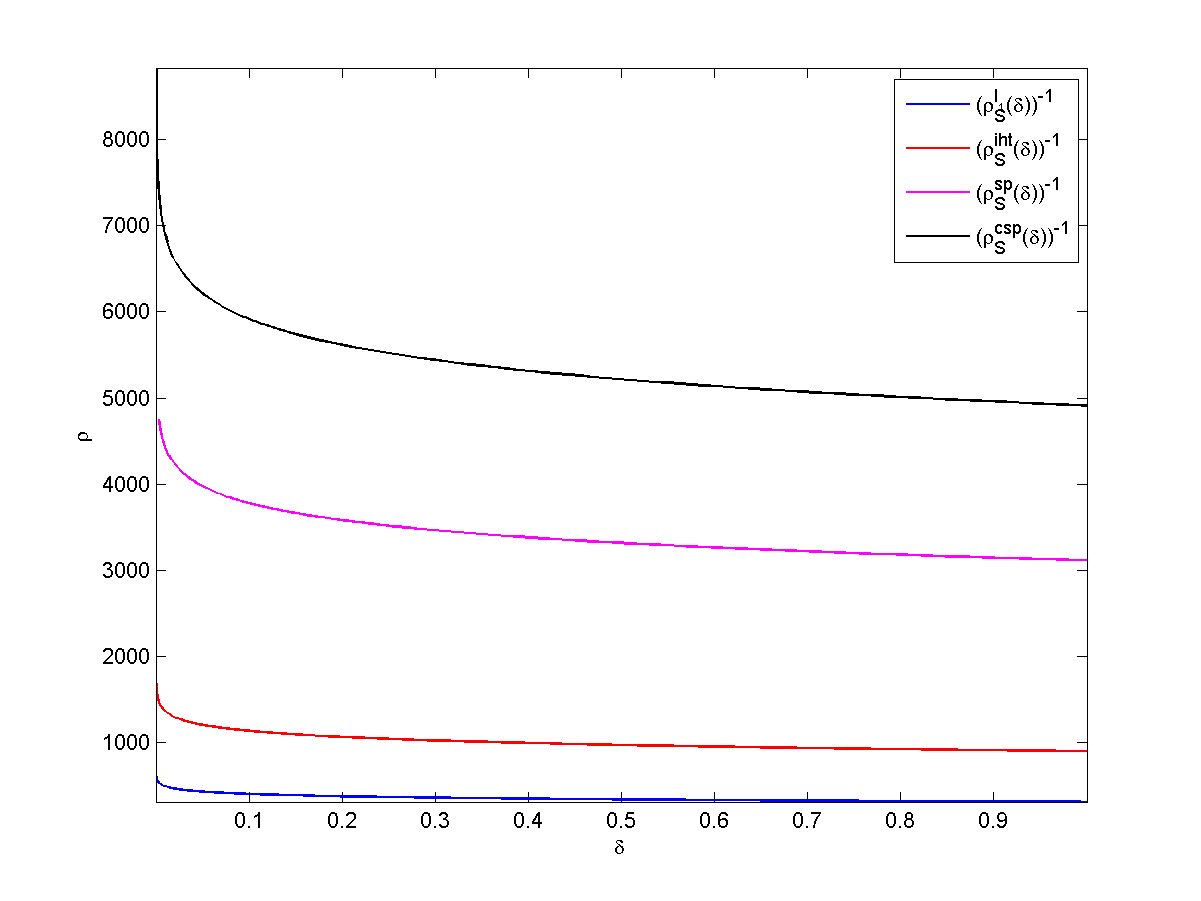}
\caption[Phase Transitions]{\textit{Left panel}: The lower bounds on the exact recovery phase transition for Gaussian random matrices for $l_{1}-\mathrm{regularization},$ IHT, SP and CoSaMP implied by the RIC bounds in Theorem \ref{BTthm}; \textit{Right panel}: The inverse of the phase transition lower bounds.}
\label{phaseT}
\end{figure}


\section{Appendix}\label{sec:appdix} Here we present the proofs of the key theorems and lemmas stated in the paper. For other theorems and lemmas, especially technical lemmas used without stating in our analysis, you are referred to the Appendix of \cite{BCT}.  

\subsection{Proof of Lemma \ref{uniqminmax}}\label{sec:Proofumm} We start by showing that $\lambda^{max}(\delta,\rho;\gamma)$ has a unique minimum for each fixed $\delta, \rho$ and $\gamma \in [\rho, \delta^{-1}]$. Equation \eqref{lmaxnew} gives the implicit relation between $\lambda^{max}$ and $\gamma$ as 
\begin{equation}
\delta \psi_{max}\left(\lambda^{max},\gamma\right) + H(\rho\delta) - \delta\gamma H\left(\rho/\gamma\right) = 0, \quad for ~\lambda^{max} \geq 1+\gamma, \nonumber
\end{equation}
where
\begin{equation}
 \psi_{max}(\lambda^{max}, \gamma) = \frac{1}{2} \big[\left(1+\gamma\right)\ln\left(\lambda^{max} \right) - \gamma \ln \gamma + 1 + \gamma - \lambda^{max} \big].  \nonumber
\end{equation}
Therefore, $\displaystyle\frac{d}{d\gamma} \left(\lambda^{max}\right) = \frac{\lambda^{max}}{\lambda^{max} - (1 + \gamma)}\ln \Bigg[\frac{\lambda^{max}\cdot(\gamma - \rho)^2}{\gamma^3} \Bigg]$ is equal to zero when
\begin{equation}
\label{dlmax}
\lambda^{max}\cdot(\gamma - \rho)^2 = \gamma^3.
\end{equation}
Let $\gamma_{min}$ satisfy (\ref{dlmax}). Since $\lambda^{max} \geq 1 + \gamma > 0$, $\frac{d}{d\gamma} \left(\lambda^{max}\right)$ is negative for $\gamma \in [\rho,\gamma_{min})$, is zero at $\gamma_{min}$ and is positive for $\gamma \in (\gamma_{min},\delta^{-1})$, equation \eqref{lmaxnew} has a unique minima over $\gamma\in [\rho,\delta^{-1}]$, and the $\gamma$ that obtains the minima is strictly greater than $\rho$.

Similarly, we show that $\lambda^{min}(\delta,\rho;\gamma)$ has a unique maximum for each fixed $\delta, \rho$ and $\gamma \in [\rho, \delta^{-1}]$.  Equation \eqref{lminnew} gives the implicit relation between $\lambda^{min}$ and $\gamma$ as 
\begin{equation}
\delta \psi_{min}\left(\lambda^{min},\gamma\right) + H(\rho\delta) - \delta\gamma H\left(\rho/\gamma\right) = 0, 
\quad for ~\lambda^{min} \leq 1-\gamma, \nonumber
\end{equation}
where
\begin{equation}
\psi_{min}\left(\lambda^{min},\gamma\right) := H\left(\gamma\right) + \frac{1}{2} \Big[ \left(1-\gamma \right)\ln\left(\lambda^{min}\right) + \gamma \ln\gamma + 1 - \gamma - \lambda^{min} \Big].  \nonumber
\end{equation}
Therefore, $\displaystyle\frac{d}{d\gamma} \left(\lambda^{min}\right) = \frac{\lambda^{min}}{(1-\gamma) - \lambda^{min}}\ln \Bigg[\frac{\gamma^3\cdot\lambda^{min}}{(1-\gamma)^2\cdot(\gamma - \rho)^2} \Bigg]$ is equal to zero when
\begin{equation}
\label{dlmin}
\gamma^3\cdot\lambda^{min} = (1-\gamma)^2(\gamma - \rho)^2.
\end{equation}
Let $\gamma_{max}$ satisfy (\ref{dlmin}).  Since $0 < \lambda^{min} \leq 1 + \gamma$, $\frac{d}{d\gamma} \left(\lambda^{min}\right)$ is positive for $\gamma \in [\rho,\gamma_{max})$, zero at $\gamma_{max}$ and negative for $\gamma \in (\gamma_{max},\delta^{-1})$, equation \eqref{lminnew} has a unique maxima over $\gamma\in [\rho,\delta^{-1}]$, and the $\gamma$ that obtains the maxima is strictly greater than $\rho$.
\begin{flushright}
    $\blacksquare$
\end{flushright}

\subsection{Proof of main results, Theorem \ref{BTthm} and Propositions \ref{propU} and \ref{propL}}\label{sec:ProofRIP} Here we give a proof similar to that given in \cite{BCT} but we take great care at the non-exponential terms necessary for the calculations of bounds of probabilities for finite values of $(k,n,N)$ in Section \ref{sec:FiniteN}. We present the proof for $\Ubt(\delta,\rho)$ in detail and sketch the proof of $\Lbt(\delta,\rho)$ which follows similarly. 

The following lemma of the bound on the probability distribution function of the maximum eigenvalue of a Wishart matrix due to Edelman, \cite{BCT,BCTT,edelman} is central to our proof.

\begin{lemma} \label{pdfmax} (\cite{edelman}, presented in this form in \cite{BCT})
Let $A_{M}$ be a matrix of size $n\times m$ whose entries are drawn i.i.d. from $\mathcal{N}(0,1/n).$ Let $f_{max}(m,n;\lambda)$ denote the distribution function for the largest eigenvalue of the derived Wishart matrix $A_M^* A_M,$ of size $m\times m$. Then $f_{max}(m,n;\lambda)$ satisfies:
\begin{equation}
\label{fmax}
f_{max}(m,n;\lambda) \leq \bigg[(2\pi)^{\frac{1}{2}} (n\lambda)^{-\frac{3}{2}} \left(\frac{n\lambda}{2}\right)^{\frac{n+m}{2}}\frac{1}{\Gamma\left(\frac{m}{2}\right) \Gamma\left(\frac{n}{2}\right)} \bigg] e^{-\frac{n\lambda}{2}} =:g_{max}(m,n;\lambda).
\end{equation}
\end{lemma}

It is helpful at this stage to rewrite Lemma \ref{pdfmax}, separating the exponential and polynomial parts (with respect to $n$) of $g_{max}(m,n;\lambda)$ as thus:

\begin{lemma} \label{pdfmax2} 
Let $\gamma_n = m/n \rightarrow \in [\rho_n,\delta_n^{-1}]$ and define 
\begin{equation}
\label{psimax}
 \psi_{max}(\lambda, \gamma) := \frac{1}{2} \Big[\left(1+\gamma\right)\ln\lambda - \gamma \ln \gamma + 1 + \gamma - \lambda \Big]. 
\end{equation}
Then
\begin{equation}
\label{fmax2}
f_{max}(m,n;\lambda)\le g_{max}(m,n;\lambda) \leq p_{max}(n,\lambda; \gamma) \exp\Big( n \cdot \psi_{max}(\lambda,\gamma)\Big) 
\end{equation}
where $p_{max}(n,\lambda; \gamma)$ is a polynomial in $n, \lambda$ and $\gamma$, given by
\begin{equation}
\label{pmax}
p_{max}(n,\lambda; \gamma) = \left(\frac{8}{\pi}\right)^{1/2} \gamma^{-1} n^{-7/2} \lambda^{-3/2}.
\end{equation} 
\end{lemma}

\begin{proof} Let $\gamma_{n} = \frac{m}{n}$ and $\frac{1}{n}\ln[g_{max}(m,n;\lambda)] = \Phi_{1}(m,n;\lambda) + \Phi_{2}(m,n;\lambda) + \Phi_{3}(m,n;\lambda)$ where 
\begin{align}
& \Phi_{1}(m,n;\lambda)  =  \frac{1}{2n}\ln(2\pi) - \frac{3}{2n}\ln(n\lambda), \quad
\Phi_{2}(m,n;\lambda)  =  \frac{1}{2} \Big[\left(1+\gamma\right)\ln(\frac{n\lambda}{2}) - \lambda \Big] \nonumber\\
& \mathrm{and} \quad \Phi_{3}(m,n;\lambda)  =  -\frac{1}{n}\ln\left(\Gamma\left(\frac{m}{2}\right) \Gamma\left(\frac{n}{2}\right)\right).\nonumber
\end{align}

\noindent We simplify $\Phi_{3}(m,n;\lambda)$ by using the second Binet's log gamma formulae \cite{WW} 
\begin{equation}
\label{loggamma}
\ln\left(\Gamma(z)\right) \ge (z - 1/2)\ln z - z + \ln \sqrt{2\pi}.
\end{equation} Thus we have
\begin{equation*}
\Phi_{2}(m,n;\lambda) + \Phi_{3}(m,n;\lambda) \le \frac{1}{2} \Big[\left(1+\gamma_{n}\right)\ln\lambda -\gamma_n\ln\gamma_n+1+\gamma_n-\lambda\Big]
-n^{-1}\ln\left(\pi\gamma_n n^2/2\right).
\end{equation*}

Incorporating $\Phi_{1}(m,n;\lambda)$ and $-n^{-1}\ln\left(\pi\gamma_n n^2/2\right)$ into $p_{max}(n,\lambda;\gamma)$ and defining the exponent (\ref{psimax}) as
\begin{equation}
\psi_{max}(\lambda, \gamma) := \mathop{\lim}_{n\rightarrow \infty}\frac{1}{n}\ln\left(g_{max}(m,n;\lambda)\right) = \frac{1}{2} \Big[\left(1+\gamma\right)\ln\lambda - \gamma\ln\gamma + \gamma + 1 - \lambda \Big] \nonumber
\end{equation}
completes the proof.
\end{proof}

The upper RIC bound $\Ubt(\delta,\rho)$ is obtained by a) construct the groups ${\cal G}_i$ according to Lemma \ref{covering},  taking a union bound over all $u=rN$ groups, and bounding the extreme eigenvalues within a group by the extreme eigenvalues of the Wishart matrices $A_{M_i}^* A_{M_i},$, see \eqref{eq:unionbound}.  In preparation for bounding the right hand side of \eqref{eq:unionbound} we compute a bound on $rNg_{max}(m,n;\lambda)$.

From Lemma \ref{pdfmax2} and equation \eqref{stirling} we have
\begin{equation}\label{asymp1}
2\lambda N{N\choose k} {m\choose k}^{-1}  g_{max}(m,n;\lambda) \le 
p'_{max}(n,\lambda) e^{n\psi_U(\lambda,\gamma)}
\end{equation}
where
\[
\psi_U(\lambda,\gamma):=\delta^{-1}\left[H(\rho\delta) - \delta\gamma H\left(\frac{\rho}{\gamma}\right) + \delta \psi_{max}(\lambda, \gamma) \right]
\]
and 
\begin{equation}
p'_{max}(n,\lambda):= 2 \lambda 
\left(\frac{5}{4}\right)^3\left(\frac{nN(\gamma-\rho)}{\gamma\delta(1-\rho\delta)}\right)^{1/2}
p_{max}(n,\lambda;\gamma).
\end{equation}

The proof of proposition \ref{propU} then follows.
\begin{proof}(Proof of Proposition \ref{propU})

For $\epsilon > 0$ with $\displaystyle \lambda^{max}(\delta,\rho) = \mathop{\min}_{\gamma} \lambda^{max}(\delta,\rho;\gamma)$ being the optimal solution to (\ref{lmaxnew}), 
\begin{eqnarray}
\label{probint}
\mathbf{P}\Big(U(k,n,N) > U(\delta_{n},\rho_{n}) + \epsilon \Big) & = & \mathbf{P}\Big(U(k,n,N) > \lambda^{max}(\delta_{n},\rho_{n}) - 1 + \epsilon \Big) \nonumber\\
& =  & \mathbf{P}\Big(1 + U(k,n,N) > \lambda^{max}(\delta_{n},\rho_{n}) + \epsilon \Big) \nonumber\\
& =  & N {N \choose k}\displaystyle {m \choose k}^{-1}  \int_{\lambda^{max}(\delta_{n},\rho_{n}) + \epsilon}^{\infty}  f_{max}(m,n;\lambda)d\lambda \nonumber\\
& \leq & N {N \choose k}\displaystyle {m \choose k}^{-1}  \int_{\lambda^{max}(\delta_{n},\rho_{n}) + \epsilon}^{\infty}  g_{max}(m,n;\lambda)d\lambda 
\end{eqnarray}
To bound the final integral in \eqref{probint} we write $g_{max}(m,n;\lambda)$ as a product of two separate functions - one of $\lambda$ and another of $n$ and $\gamma_n$, as $g_{max}(m,n;\lambda) = \varphi(n, \gamma_{n})\lambda^{-\frac{3}{2}} \lambda^{\frac{n}{2}(1 + \gamma_{n})e^{-\frac{n}{2}\lambda}}$ where 
$$\varphi(n, \gamma_{n}) = (2\pi)^{\frac{1}{2}} (n)^{-\frac{3}{2}} \left(\frac{n}{2}\right)^{\frac{n}{2}\left(1 + \gamma_{n} \right)}\frac{1}{\Gamma\left(\frac{n}{2}\gamma_{n}\right) \Gamma\left(\frac{n}{2}\right)}. $$

\noindent With this and using the fact that $\lambda^{max}(\delta_{n},\rho_{n}) > 1 + \gamma_{n}$ and that $\lambda^{\frac{n}{2}(1 + \gamma_{n})}e^{-\frac{n}{2}\lambda}$ is strictly decreasing in $\lambda$ on $[\lambda^{max}(\delta_{n},\rho_{n}),\infty)$ we can bound the integral in (\ref{probint}) as follows.

\begin{eqnarray}
\label{intgmax}
\displaystyle \int_{\lambda^{max}(\delta_{n},\rho_{n}) + \epsilon}^{\infty}  g_{max}(m,n;\lambda)d\lambda & \leq & \varphi(n, \gamma_{n}) \displaystyle {[\lambda^{max}(\delta_{n},\rho_{n}) + \epsilon]}^{\frac{n}{2}(1 + \gamma_{n})} e^{-\frac{n}{2}(\lambda^{max}(\delta_{n},\rho_{n}) + \epsilon)} \int_{\lambda^{max}(\delta_{n},\rho_{n}) + \epsilon}^{\infty} \lambda^{-\frac{3}{2}} d\lambda  \nonumber\\
& = & \displaystyle{[\lambda^{max}(\delta_{n},\rho_{n})]}^{\frac{3}{2}} g_{max}\big[m,n;\lambda^{max}(\delta_{n},\rho_{n}) + \epsilon\big] \int_{\lambda^{max}(\delta_{n},\rho_{n}) + \epsilon}^{\infty} \lambda^{-\frac{3}{2}} d\lambda \nonumber\\
& = & 2 \displaystyle{\lambda^{max}(\delta_{n},\rho_{n})}  g_{max}\big[m,n;\lambda^{max}(\delta_{n},\rho_{n}) + \epsilon\big] 
\end{eqnarray}

\noindent Thus (\ref{probint}) and (\ref{intgmax}) together gives
\begin{eqnarray}
\label{probfinal}
\mathbf{P}\Big(U(k,n,N) > U(\delta_{n},\rho_{n}) + \epsilon \Big) & \leq & 2 {\lambda^{max}(\delta_{n},\rho_{n})} r N g_{max}\big[m,n;\lambda^{max}(\delta_{n},\rho_{n}) + \epsilon\big]  \nonumber\\
& \leq & p'_{max}\Big(n,\lambda^{max}(\delta_{n},\rho_{n})\Big) \exp\big[n \cdot \psi_{U}\left(\lambda^{max}(\delta_{n},\rho_{n}) + \epsilon\right)\big]\nonumber\\
& \leq & p'_{max}\Big(n,\lambda^{max}(\delta_{n},\rho_{n})\Big) \exp\bigg[n\epsilon \cdot \frac{d}{d\lambda} \psi_{U}\Big(\lambda^{max}(\delta_{n},\rho_{n})\Big)\bigg],
\end{eqnarray}
where $r = {N \choose k} {m \choose k}^{-1}$, and the last inequality is due to $\psi_{U}(\lambda)$ being strictly concave.
\end{proof}

The following is a corollary to Proposition \ref{propU}:
\begin{corollary}
\label{corlpropU} Let $(\delta,\rho) \in (0,1)^2$ and let $A$ be a matrix of size $n\times N$ whose entries are drawn i.i.d. from $\mathcal{N}(0,1/n).$ Define $\Ubt(\delta,\rho) = \lambda^{max}(\delta,\rho) - 1$ where $\lambda^{max}(\delta,\rho;\gamma)$ is the solution of (\ref{lmaxnew}) for each $\gamma \in [\rho, \delta^{-1}]$ and $\displaystyle \lambda^{max}(\delta,\rho) := \mathop{\min}_{\gamma} \lambda^{max}(\delta,\rho;\gamma)$. Then for any $\epsilon > 0,$ in the proportional-growth asymptotics
$$\mathbf{P}\left(U(k,n,N) > \Ubt(\delta,\rho) + \epsilon \right) \rightarrow 0 $$ exponentially in $n$.
\end{corollary}

\begin{proof}
From (\ref{probfinal}), since $\frac{d}{d\lambda}\psi_{U} \left(\lambda^{max}(\delta_{n},\rho_{n})\right) < 0$ is strictly bounded away from zero and all the limits of $(\delta_{n},\rho_{n})$ are smoothly varying functions we conclude, for any $\epsilon > 0$ 
$$\lim_{n \rightarrow \infty} \mathbf{P}\left(U(k,n,N) > \Ubt(\delta,\rho) + \epsilon \right) \rightarrow 0.$$ 
\end{proof}

\noindent Thus we finish the proof for $\Ubt(\delta,\rho)$. We sketch the similar proof for Proposition \ref{propL} and $\Lbt(\delta,\rho)$.  Bounds on the probability distribution function of the minimum eigenvalue of a Wishart matrix are given in the following lemma.

\begin{lemma} \label{pdfmin} (\cite{edelman}, presented in this form in \cite{BCT})
Let $A_{M}$ be a matrix of size $n\times m$ whose entries are drawn i.i.d. from $\mathcal{N}(0,1/n).$ Let $f_{min}(m,n;\lambda)$ denote the distribution function for the smallest eigenvalue of the derived Wishart matrix $A_M^* A_M,$ of size $m\times m$. Then $f_{min}(m,n;\lambda)$ satisfies:

\begin{equation}
\label{fmin}
f_{min}(m,n;\lambda) \leq (\frac{\pi}{2n\lambda})^{\frac{1}{2}}  \left(\frac{n\lambda}{2}\right)^{\frac{n-m}{2}} \bigg[\frac{\Gamma\left(\frac{n+1}{2}\right)}{\Gamma\left(\frac{m}{2}\right) \Gamma\left(\frac{n-m+1}{2}\right) \Gamma\left(\frac{n-m+2}{2}\right)} \bigg] e^{-\frac{n\lambda}{2}}  =: g_{min}(m,n;\lambda)
\end{equation}
\end{lemma}
Again an explicit expression of $f_{min}(m,n;\lambda)$ in terms of exponential and polynomial parts leads to the following Lemma.

\begin{lemma} \label{pdfmin2} 
Let $\gamma_n = m/n$ and define
\begin{equation}
\label{psimin}
\psi_{min}(\lambda, \gamma) := H\left(\gamma\right) + \frac{1}{2} \Big[\left(1-\gamma\right)\ln\lambda + \gamma \ln\gamma + 1 - \gamma - \lambda   \Big].
\end{equation}
Then
\begin{equation}
\label{fmin2}
f_{min}(m,n;\lambda) \le g_{min}(m,n;\lambda) \leq p_{min}(n,\lambda) \exp \Big(n \cdot \psi_{max}(\lambda,\rho,\delta;\gamma) \Big) 
\end{equation}
where $p_{min}(n,\lambda; \gamma)$ is a polynomial in $n, \lambda$ and $\gamma$, given by
\begin{equation}
\label{pmin}
p_{min}(n,\lambda; \gamma) = \frac{e}{2\pi\sqrt{2\lambda}}.
\end{equation} 
\end{lemma}

The proof of Lemma \ref{pdfmin2} follows that of Lemma \ref{pdfmax2} and is omitted for brevity.  Equipped with Lemma \ref{pdfmin2} a large deviation analysis yields
\begin{equation}
\label{asymp2}
2\lambda N{N\choose k} {m\choose k}^{-1}  g_{max}(m,n;\lambda) \le p'_{min}(n,\lambda) e^{n\psi_L(\lambda,\gamma)}
\end{equation}
where
\[ \psi_L(\lambda,\gamma):=\delta^{-1}\left[H(\rho\delta) - \delta\gamma H\left(\frac{\rho}{\gamma}\right) + \delta \psi_{min}(\lambda, \gamma) \right], \]
and
\begin{equation}
p'_{min}(n,\lambda):=2\lambda\left(\frac{5}{4}\right)^3  
\left(\frac{nN(\gamma-\rho)}{\gamma\delta(1-\rho\delta)}\right)^{1/2} p_{min}(n,\lambda).
\end{equation}

With Lemma \ref{pdfmin2} and (\ref{asymp2}), Proposition \ref{propL} follows similarly to the proof of Proposition \ref{propU} stated earlier in this section.  The bound $\Lbt(\delta,\rho)$ is a corollary of Proposition \ref{propL}.
\begin{corollary}
\label{corlpropL} Let $(\delta,\rho) \in (0,1)^2$ and let $A$ be a matrix of size $n\times N$ whose entries are drawn i.i.d. from $\mathcal{N}(0,1/n).$ Define $\Lbt(\delta,\rho) := 1 - \lambda_{min}(\delta,\rho)$ where $\lambda_{min}(\delta,\rho;\gamma)$ is the solution of (\ref{lminnew}) for each $\gamma \in [\rho, \delta^{-1}]$ and $\displaystyle \displaystyle \lambda^{min}(\delta,\rho) := \mathop{\min}_{\gamma} \lambda^{min}(\delta,\rho;\gamma)$. Then for any $\epsilon > 0,$ in the proportional-growth asymptotic
$$\mathbf{P}\left(L(k,n,N) > \Lbt(\delta,\rho) + \epsilon \right) \rightarrow 0$$ exponentially in $n$.
\end{corollary}

\begin{flushright} $\blacksquare$\end{flushright}

\bibliographystyle{plain}
\bibliography{improved_rip_20100314}

\end{document}